\newif\ifcomm
\commtrue

\newcommand{\CLASSINPUToutersidemargin}{0.625in}

\documentclass[10pt,journal,twocolumn]{IEEEtran}

\makeatletter
\def\ps@headings{%
	\def\@oddhead{\mbox{}\scriptsize\rightmark \hfil \thepage}%
	\def\@evenhead{\scriptsize\thepage \hfil\leftmark\mbox{}}%
	\def\@oddfoot{}%
	\def\@evenfoot{}}
\makeatother
\usepackage{booktabs} 
\usepackage{balance}
\usepackage{xspace}
\usepackage{color}
\usepackage{cite}

\usepackage{amsfonts}
\usepackage{amsmath,amssymb}
\usepackage{amsthm}
\usepackage{diagbox}
\usepackage[boxruled]{algorithm2e}
\usepackage{bbm}
\usepackage{graphicx}
\usepackage{subcaption}
\usepackage{url}
\usepackage{xcolor}

\usepackage[normalem]{ulem}

\usepackage{array,multirow,graphicx}
\usepackage{float}

\usepackage{booktabs}  

\usepackage{array}
\newcolumntype{C}[1]{>{\centering\let\newline\\\arraybackslash\hspace{0pt}}m{#1}}

\newcommand{\mycc}{1.0cm}

\newcommand{\T}[1]{\smallskip\noindent\textbf{#1}}
\newcommand{\TT}[1]{\noindent\textbf{#1}}
\newcommand{\I}[1]{\smallskip\noindent\textit{#1}}

\newcommand{\set}[1]{\left\{#1\right\}}
\newcommand{\bp}[1]{\Big(#1\Big)}
\newcommand{\given}[1]{\,\Big\vert\,#1}

\providecommand{\ie}{\emph{i.e.,} }
\providecommand{\eg}{\emph{e.g.,} }
\providecommand{\vs}{vs.\ }

\providecommand{\name}{$LSQ$\xspace}

\DeclareMathOperator{\E}{\mathbb{E}}

\newtheorem{theorem}{Theorem}
\newtheorem{lemma}{Lemma}

\newtheorem{definition}{Definition}
\newtheorem{example}{Example}
\newtheorem{proposition}{Proposition}

\newtheorem{assumption}{Assumption}

\ifcomm
	\newcommand{\mycomm}[3]{{\footnotesize{{\color{#2} \textbf{[#1: #3]}}}}}
    \newcommand{\Fmycomm}[3]{\footnote{{{\color{#2} \textbf{[#1: #3]}}}}}
\else
    \newcommand{\mycomm}[3]{}
    \newcommand{\Fmycomm}[3]{}
\fi

\newcommand{\brac}[1]{\left\{ #1 \right\}}
\newcommand{\sbrac}[1]{\left[ #1 \right]}

\newcommand{\new}[1]{\textcolor{black}{#1}}
\newcommand{\newer}[1]{\textcolor{black}{#1}}
\newcommand{\newest}[1]{\textcolor{black}{#1}}

\begin{document}

\title{LSQ: Load Balancing in Large-Scale Heterogeneous Systems with Multiple Dispatchers}

\author{Shay Vargaftik,
        Isaac Keslassy, 
        and Ariel Orda 

\thanks{This manuscript is an extended version of a paper accepted for publication in the IEEE/ACM TRANSACTIONS ON NETWORKING (submitted on 15 Apr 2019, revised on 12 Aug 2019, 19 Nov 2019, 09 Jan 2020 and accepted on 26 Feb 2020). It is based on part of the Ph.D. thesis of Shay Vargaftik \cite{shayvthesys}, submitted to the Senate of the Technion on Apr 2019. Preprint of this work has been available online on \cite{Openreviewlsq} since 27 Dec 2018.}

\thanks{Shay Vargaftik is with VMware Research, E-mail: shayv@vmware.com. Isaac Keslassy and Ariel Orda are with the Technion, E-mail: \{isaac@ee, ariel@ee\}.technion.ac.il.}

\thanks{© 2020 IEEE.  Personal use of this material is permitted.  Permission from IEEE must be obtained for all other uses, in any current or future media, including reprinting/republishing this material for advertising or promotional purposes, creating new collective works, for resale or redistribution to servers or lists, or reuse of any copyrighted component of this work in other works.}

}

\maketitle

\begin{abstract}

Nowadays, the efficiency and even the feasibility of traditional load-balancing policies are challenged by the rapid growth of cloud infrastructure and the increasing levels of server heterogeneity. In such  heterogeneous systems with many load-balancers, traditional solutions, such as $JSQ$, incur a prohibitively large communication overhead and detrimental incast effects due to herd behavior. Alternative low-communication policies, such as $JSQ(d)$ and the recently proposed $JIQ$, are either unstable or provide poor performance.

We introduce the \textit{Local Shortest Queue (\name)} family of load balancing algorithms. In these algorithms, each dispatcher maintains its own, local, and possibly outdated view of the server queue lengths, and keeps using $JSQ$ on its local view. A small communication overhead is used infrequently to update this local view. We formally prove that as long as the error in these local estimates of the server queue lengths is bounded  \emph{in expectation}, the entire system is strongly stable. Finally, in simulations, we show how simple and stable \name{} policies exhibit appealing performance and significantly outperform existing low-communication policies, while using an equivalent communication budget. In particular, our simple policies often outperform even $JSQ$ due to their reduction of herd behavior. We further show how, by relying on smart servers (\ie advanced pull-based communication), we can further improve performance and lower communication overhead. 

\end{abstract}


\section{Introduction}

\TT{Background.} In recent years, due to the rapidly increasing size of cloud services and applications~\cite{foster2008cloud,nishtala2013scaling,dean2008mapreduce,gupta2007analysis}, the design of load balancing algorithms for parallel server systems has become extremely challenging. The goal of these algorithms is to efficiently load-balance incoming jobs to a large number of servers, even though these servers display large \textit{heterogeneity} for two reasons.
First, current large-scale systems increasingly contain, in addition to multiple generations of CPUs (central processing units)~\cite{govindan2016evolve}, various types of accelerated devices such as GPUs (graphics processing units), FPGAs (field-programmable gate arrays) and ASICs (application-specific integrated circuit), with significantly higher processing speeds. Second, VMs (virtual machines) or containers are commonly used to deploy different services that \new{may} share resources on the same servers, potentially leading to significant and unpredictable heterogeneity~\new{\cite{kannan2018proctor,garg2018migrating}}.

In a traditional server farm, a centralized load-balancer (dispatcher) can rely on a full-state-information policy with strong theoretical guarantees for heterogeneous servers, such as join-the-shortest-queue ($JSQ$), which routes emerging jobs to the server with the shortest queue~\cite{winston1977optimality,weber1978optimal,foschini1978basic,gupta2007analysis,foley2001join}. This is because in such single-centralized-dispatcher scenarios, the dispatcher forms a single access point to the servers. Therefore, by merely receiving a notification from each server upon the completion of each job, it can track all queue lengths, because it knows the exact arrival and departure patterns of each queue (neglecting propagation times)~\cite{lu2011join}. 
The communication overhead between the servers and the dispatcher is at most a single message per job, which is appealing and does not increase with the number of servers. 
\new{Note that unless Direct Server Return (DSR) is employed, there is not even a need for this additional message per job, since all job responses return through the single dispatcher anyway.}

However, in current clouds, which keep growing in size and thus have to rely on multiple dispatchers \cite{gandhi2015duet}, implementing a policy like $JSQ$ may incur two main problems. (1)~It involves a prohibitive implementation and communication overhead as the number $m$ of dispatchers increases \cite{lu2011join}; this is because each server needs to keep all $m$ dispatchers updated as jobs arrive and complete, leading to $O(m)$ communication messages per job \new{(and this still holds even when DSR is not employed, since any reply only transits through a single dispatcher)}. (2)~Also, it may suffer from incast issues when all/many dispatchers send at once all incoming traffic to the currently-shortest queue.  These two problems force cloud dispatchers to rely on policies that do not provide any service guarantees with multiple dispatchers and heterogeneous servers \cite{citrix,youtube_lec}. For instance, two widely-used open-source load balancers, namely HAProxy and NGINX, have recently introduced the ``power of two choices'' ($JSQ(2)$) policy into their L7 load-balancing algorithms \cite{ngynx_po2,haproxy_po2}.\footnote{Quoting \cite{ngynx_po2} (by Owen Garrett, Head of Products at NGINX): ``\emph{Classic load-balancing methods such as Least Connections [$JSQ$] work very well when you operate a single active load balancer which maintains a complete view of the state of the load-balanced nodes. The ``power of two choices'' approach is not as effective on a single load balancer, but it deftly avoids the bad-case ``herd behavior'' that can occur when you scale out to a number of independent load balancers. This scenario is not just observed when you scale out in high-performance environments; it's also observed in containerized environments where multiple proxies each load balance traffic to the same set of service instances}.''}

\T{Related work.} Despite their increasing importance, scalable policies for heterogeneous systems with multiple dispatchers have received little attention in the literature. In fact, as we later discuss, the only suggested scalable policies that address the many-dispatcher scenario in a heterogeneous setting are based on join-the-idle-queue ($JIQ$) schemes, and none of them is stable \cite{zhou2017designing}. 

In the $JSQ(d)$ (power-of-choice) policy, to make a routing decision, a dispatcher samples $d \geq 2$ queues uniformly at random and chooses the shortest among them~\cite{maguluri2014heavy,gamarnik2018delay,mukhopadhyay2016randomized,ying2017power,bramson2012asymptotic,bramson2010randomized}.  $JSQ(d)$ is stable in systems with homogeneous servers. However, with heterogeneous servers, $JSQ(d)$ leads to poor performance and even to instability, both with a single as well as with multiple dispatchers~\cite{foss1998stability}. 

In the $JSQ(d,m)$ (power-of-memory) policy, the dispatcher samples the $m$ shortest queues from the previous decision in addition to  $d \geq m \geq 1$ new queues chosen uniformly-at-random~\cite{shah2002use,mitzenmacher2002load}. The job is then routed to the shortest among these $d+m$ queues. $JSQ(d,m)$ has been shown to be stable in the case of a single dispatcher when $d=m=1$, even with heterogeneous servers. However, it offers poor performance in terms of job completion time, and it has not been studied in the multiple-dispatcher realm, thus has no theoretical guarantees.

A recent study \cite{zhou2017designing} proposes a class of policies that are both throughput optimal and heavy-traffic delay optimal. However, their assumptions are not aligned with our system model and motivation due to several reasons: (1) For heterogeneous servers, \cite{zhou2017designing} requires the knowledge of the server service rates, which may not be achievable in practice. (2) \cite{zhou2017designing} assumes that the number of jobs that a server may complete in a time slot, as well as the number of jobs that may arrive at a dispatcher in a time slot, are deterministically upper-bounded, which rules out important modeling options with unbounded support, such as geometric services or Poisson arrivals. (3) Most importantly, they consider only a single dispatcher, and it is unclear whether their analysis and performance guarantees can be extended to multiple dispatchers.

In addition, to address the communication overhead in systems with multiple dispatchers, the $JIQ$ policy has been proposed \cite{lu2011join,mitzenmacher2016analyzing,stolyar2017pull,stolyar2015pull,van2017load}. Roughly speaking, in $JIQ$, each dispatcher routes jobs to an idle server, if it is aware of any, and to a random server otherwise. Servers may only notify dispatchers when they become idle. $JIQ$ achieves low communication overhead of at most a single message per job, irrespective of the number of dispatchers, and good performance at low and moderate loads when servers are homogeneous \cite{lu2011join}. However, for heterogeneous servers, $JIQ$ is not stable, \ie it fails to achieve 100\% throughput \cite{zhou2017designing}.

Finally, two recent studies on low-communication load balancing \cite{jonatha2018power,van2019hyper} propose to use local memory as well to hold the possibly-outdated server states. As we later show, these policies are, in fact, special cases of \name. That is, they only consider a single dispatcher and homogeneous servers while we consider multiple dispatchers and heterogeneous servers, which hold the main motivation and contribution of our work. \new{Moreover, there are additional significant differences between our model assumptions and theirs that affect the analysis. For example, they consider a continuous-time model with Poisson arrivals and exponential service rates in which incast is impossible, whereas our model is in discrete-time and we only assume the existence of a first and a second moment of the processes. In particular, we do not assume any specific distribution of the arrivals or service rates by the servers. Also, they analyze their algorithms in a large-system limit, whereas our analysis deals with a finite number of servers and dispatchers.}

\T{Contributions.} This paper makes the following contributions:

\I{Local Shortest Queue (\name{}).} We introduce \name{}, a new family of load balancing algorithms for large-scale heterogeneous systems with multiple dispatchers. As Figure \ref{fig:model} illustrates, in \name{}, each dispatcher keeps a local view of the server queue lengths and routes jobs to the shortest among them. Communication overhead among the servers and the dispatchers is used only to update the local views and make sure they are not too far from the real server queue lengths. 

\begin{figure*}[t!]
\centering
\begin{subfigure}{0.35\linewidth}
\includegraphics[width=\textwidth]{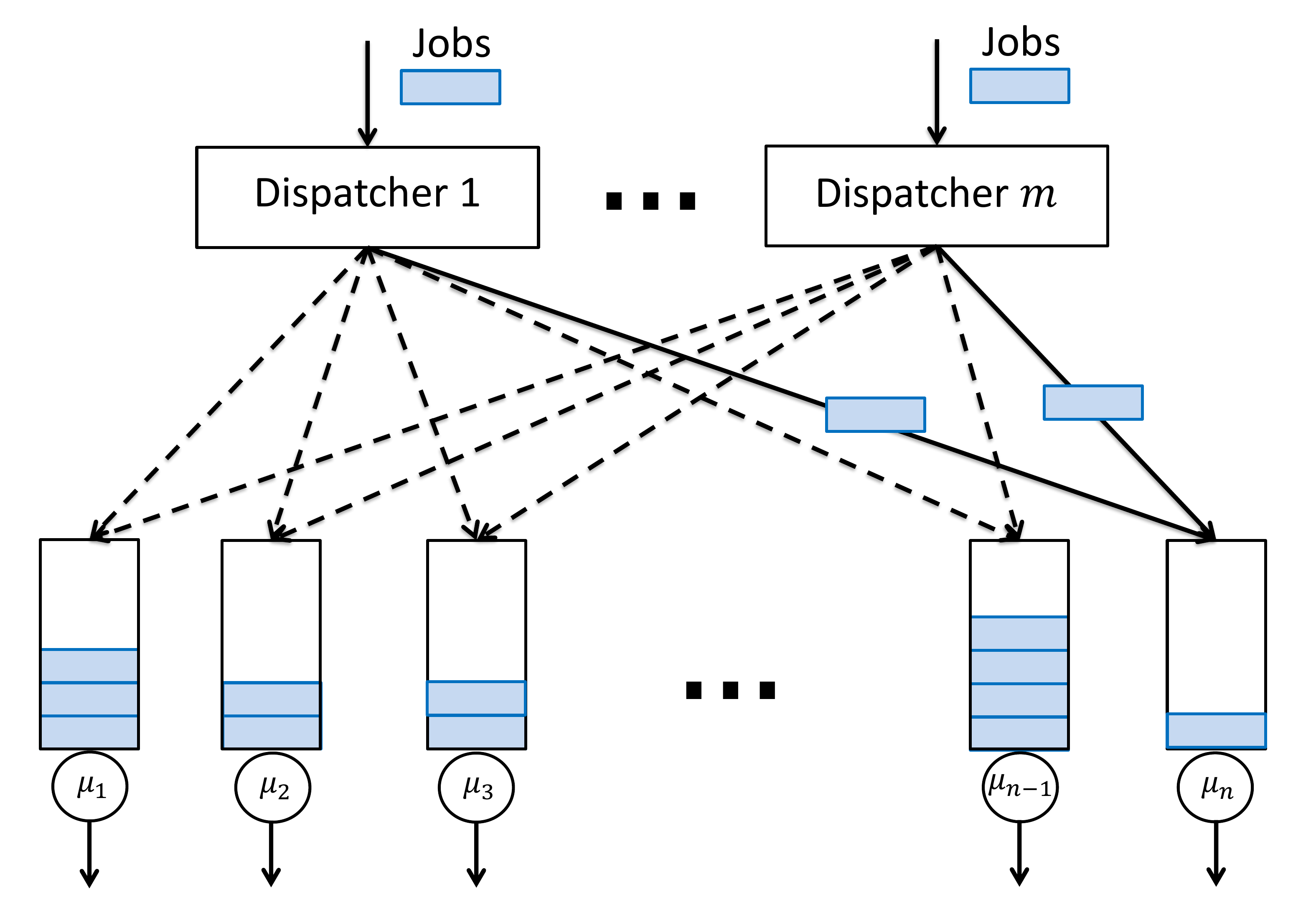}
  \caption{The $JSQ$ approach.}
  \label{fig:model:global}
\end{subfigure}
\quad\quad\quad\quad\quad
\begin{subfigure}{0.35\linewidth}
  \includegraphics[width=\textwidth]{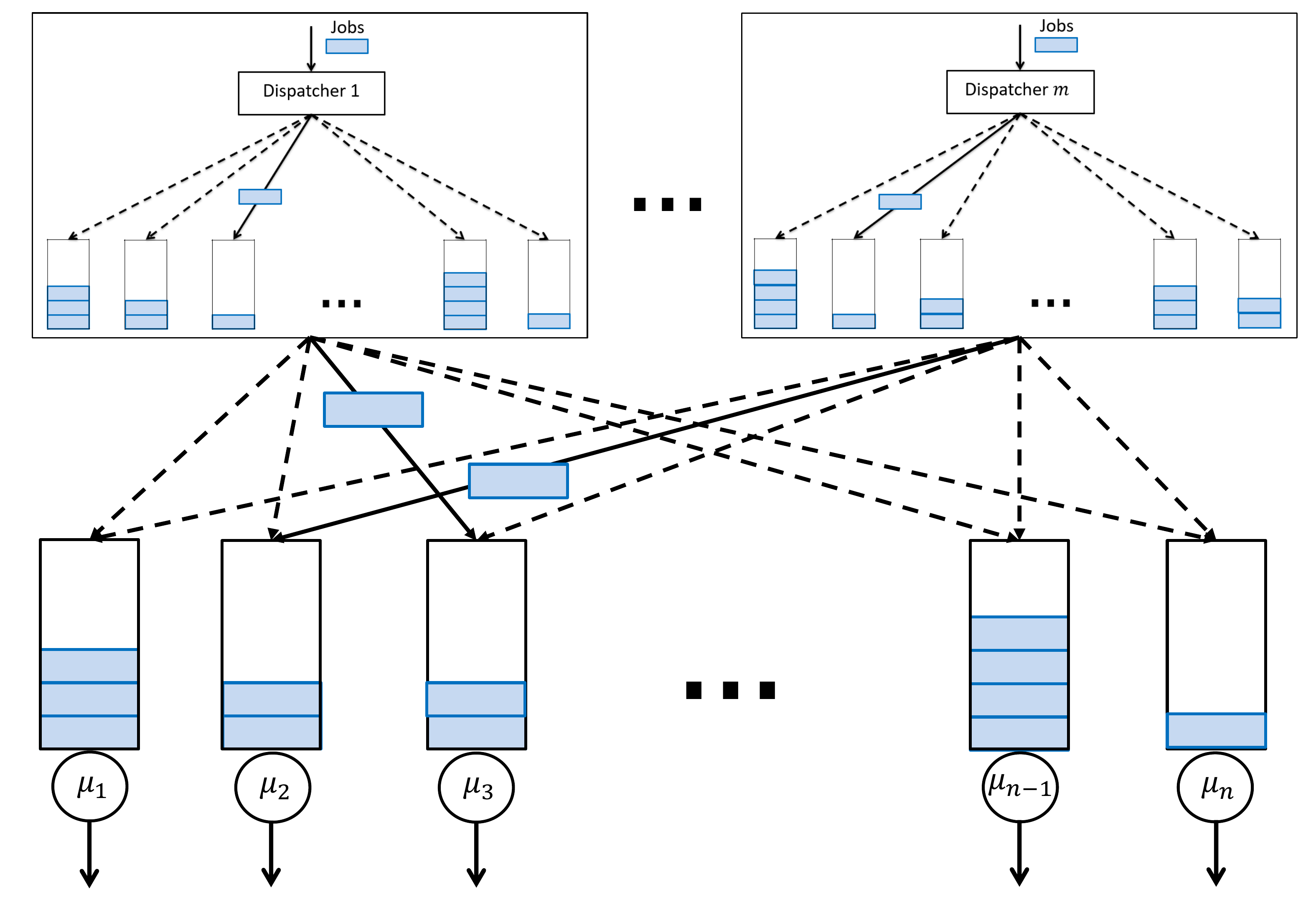}
  \caption{The \name{} approach.}
  \label{fig:model:local}
\end{subfigure}
\caption{$JSQ$ \vs \name. \textbf{(a)} $JSQ$ requires instant knowledge of all server queues by all dispatchers, resulting in substantial communication overhead and possible bad-case ``herd-behaviour'' leading to incast issues. \textbf{(b)} At each dispatcher, \name{} relies on limited current and past information from the servers to construct a local view of all the server queue lengths. For instance, dispatcher 1 believes that the queue length at server 3 is 1, while it is 2. It then sends jobs to the shortest queue as dictated by its view (here, to server 3 rather than server $n$). Communication is used only to update the local views of the dispatchers, \ie to improve their local views.}
\label{fig:model}
\end{figure*}

\I{Sufficient stability condition and stability proof.} We prove that all \name{} policies that keep a \emph{bounded distance in expectation} between the real queue lengths and the local views are strongly stable, \ie keep bounded expected queue lengths. The main difficulty in the proof arises from the fact that the decisions taken by an \name{} policy depend on the \emph{local view of each dispatcher}, hence on a potentially long history of system states. To address this challenge, we introduce two additional {stable} policies into our analysis: (1)~$JSQ$ and (2)~Weighted-Random ($WR$). Roughly speaking, we show that our policy is sufficiently similar to $JSQ$ which, in turn, is better than $WR$. We complete the proof by using the fact that, unlike $JSQ$,  $WR$ takes routing decisions that do not depend on the system state. 

\I{Simplified stability conditions.} It can be challenging to prove that an \name{} policy is stable, \ie that in expectation, the local dispatcher views are not too far from the real queue lengths. Therefore, we develop simpler sufficiency conditions to prove that an \name{} policy is stable and exemplify their use.  

\I{Stable \name{} policies.} Since \name{} is not restricted to work with either push-- (\ie dispatchers sample the servers) or pull-- (\ie servers update the dispatchers) based communication, we aim to achieve the same communication overhead as the lowest-overhead/best-known examples in each class. Accordingly, we show how to construct new stable \name{} policies with communication patterns similar to those of other low-communication policies such as the push-based $JSQ(2)$, but with significantly stronger theoretical guarantees. 

\I{Simulations.} Using simulations we show how simple and stable \name{} policies present appealing performance, and significantly outperform other low-communication policies using an equivalent communication budget. Our simple policies often outperform even $JSQ$. This is achieved by sending jobs to less loaded servers but also by reducing herd behavior when compared to $JSQ$, as different dispatchers have different views. 

\I{Smart servers.} We show how relying on smart servers (\ie advanced pull-based communication) allows us to improve performance and communication overhead even further and to \emph{consistently} outperform $JSQ$ in terms of both mean queues lengths and job completion time delay tail distribution. We rely on two main elements to achieve this: (1)~fine-tuning the probabilities at which servers send messages, such that less loaded servers send messages with a higher probability; (2)~when a message is sent by a server, it is sent to the dispatcher with the worst local view of this server. 

\I{Simulation code.} To benefit the research community, we made our evaluation code available online\cite{lsqcode}.

\section{System model}

We consider a system with a set $M = \set{1,2,\ldots,m}$ of dispatchers load-balancing incoming jobs among a set $N = \set{1,2,\ldots,n}$ of possibly-heterogeneous servers. 

\T{Time slots.} We assume a time slotted system with the following order of events within each time slot: (1)~jobs arrive at each dispatcher; (2)~a routing decision is taken by each dispatcher and it immediately forwards its jobs to one of the servers; (3)~each server performs its service for this time-slot. 

Note that, for simplicity and ease of exposition, we assume that all dispatcher time slots are synchronized; \newer{in practice, such synchronization can be achieved by well established techniques such as the Precision Time Protocol (PTP) \cite{ptp} that is commonly used in order to establish sub-microsecond network-wide time synchronization. For instance, \cite{perry2014fastpass} relies on NIC-based synchronization, which achieves a time accuracy of a few hundred nanoseconds even across a multi-hop network.}

Nonetheless, we believe that our results may be extended to a framework where time slots are not synchronized among the dispatchers, and leave  such a generalization to future work.

\T{Dispatchers.} As mentioned, each of the $m$ dispatchers does not store incoming jobs, and instead immediately forwards them to one of the $n$ servers. We denote by $a^j(t)$ the number of exogenous job arrivals at dispatcher $j$ at the beginning of time slot $t$. We make the following assumption:
\begin{equation}\label{eq:arrival_1}
    \brac{a(t) = \sum_{j=1}^m a^j(t)}_{t=0}^{\infty} \text{ is an $i.i.d.$ process}
\end{equation}
\begin{equation}\label{eq:arrival_2}
    \E[a(0)] = \lambda^{(1)}
\end{equation}
\begin{equation}\label{eq:arrival_3}
    \E \Big[\bp{a(0)}^2\Big] = \lambda^{(2)}
\end{equation} 
That is, we only assume that the total job arrival process to the system is $i.i.d.$ over time slots and admits finite first and second moments. \new{Note that we do not assume any specific process or any deterministic bound on the number of arrived jobs at a given time slot}. The division of arriving jobs among the dispatchers is assumed to follow any arbitrary policy that does not depend on the system state (\ie queue lengths).  Furthermore, we just assume that there is a positive probability of job arrivals at all dispatchers. That is, we assume that there exists a strictly positive constant $\epsilon_0$ such that 
\begin{equation}\label{eq:arrival_4}
    \mathbbm{P}(a^j(t)>0)>\epsilon_0 \quad \forall (j,t) \in M \times \mathbb{N}.
\end{equation}
This, for example, covers complex scenarios with time-varying arrival rates to the different dispatchers that are not necessarily independent. We are not aware of previous work covering such general scenarios \new{with possibly correlated arrivals at the different dispatchers}. 

We further denote by $a_i^j(t)$ as the number of jobs forwarded by dispatcher $j$ to server $i$ at the beginning of time slot $t$. Let $$a_i(t) = \sum_{j=1}^m a_i^j(t)$$ be the total number of jobs forwarded to server $i$ by all dispatchers at time slot $t$. \newer{Finally, we assume that the arrival and service rates are unknown to the dispatchers.}

\T{Servers.} Each server has a FIFO queue for storing incoming jobs. Let $Q_i(t)$ be the queue length of server $i$ at the beginning of time slot $t$ (before any job arrivals and departures at time slot $t$). We denote by $s_i(t)$ the potential service offered to queue $i$ at time slot $t$. That is, $s_i(t)$ is the maximum number of jobs that can be completed by server $i$ at time slot $t$. We assume that, for all $i \in N$,  
\begin{equation}\label{eq:service_1}
    \{s_i(t)\}_{t=0}^{\infty} \text{ is $i.i.d.$ over time slots}
\end{equation}
\begin{equation}\label{eq:service_2}
    \E[s_i(0)] = \mu_i^{(1)}
\end{equation}
\begin{equation}\label{eq:service_3}
    \E \Big[\bp{s_i(0)}^2\Big] = \mu_i^{(2)}
\end{equation}
Namely, we assume that the service process of each server is $i.i.d.$ over time slots and admits finite first and second moments. \new{Note that we do not assume any specific process or any deterministic bound on the number of completed jobs at a given time slot}. We also assume that all service processes are mutually independent across the different servers and, furthermore, they are independent of the arrival
processes. 

\T{Admissibility.} We assume the system is sub-critical, i.e., that there exists an $\epsilon > 0$ such that
\begin{equation}\label{eq:subcritical_assumption}
    \sum_{i=1}^n \mu_i^{(1)} - \lambda^{(1)} = \epsilon. 
\end{equation}

\section{\name load balancing}
Next, we formally introduce the \name{} family of load balancing policies. Then, we introduce our main theoretical result of the paper: namely, we establish a sufficient, easy to satisfy, condition for an \name{} policy to be stable. 

\subsection{The \name{} Family}
We assume that each dispatcher $j \in M$ holds a local view of each server's $i \in N$ queue length. We denote by $\tilde{Q}_i^j(t)$ the queue length of server $i$ as dictated by the local view of dispatcher $j$ at the beginning of time slot $t$ (before any arrivals and departures at that time slot). Finally, we can define \name.
\begin{definition}[Local Shortest Queue (\name)]
We term a load balancing policy as an \name{} policy iff at each time slot, each dispatcher $j$ follows the $JSQ$ policy based on its local view of the queue lengths, \ie $\{\tilde{Q}_i^j(t)\}_{i=1}^n$. That is, dispatcher $j$ forwards all of its incoming jobs at the beginning of time slot $t$ to a server $i^*$ such that $i^* \in argmin_i \{\tilde{Q}_i^j(t)\}_{i=1}^n$ (ties are broken randomly).
\end{definition}

As we later show, this broad definition provides appealing flexibility when designing a load balancing policy, i.e., there are numerous approaches for how to \emph{update the local views} of the dispatchers.  

\subsection{Sufficient stability condition}

We proceed to introduce a sufficient condition for an \name{} policy to be stable. This condition essentially states that the difference between the local views of the dispatchers and the real queue lengths of the servers should be \emph{bounded in expectation}. \new{Note that the actual difference between the local views and the real queue states may be unbounded. As we later discuss, this loose assumption allows flexibility in the algorithm design with strong theoretical guarantees and appealing performance}. Formally:

\begin{assumption}\label{asm_2}
 There exists a constant $C>0$ such that at the beginning of each time slot (before any arrivals and departures), it holds that
\begin{equation}\label{eq:disp_knowledge}
    \E \Big[ \big|Q_i(t) - \tilde{Q}_i^j(t)\big| \Big] \le C \quad \forall \, (i,j,t) \in N \times M \times \mathbb{N}. 
\end{equation}
\end{assumption}
We will later rely on it to prove the  main theoretical result of this paper, \ie that  any \name{} load balancing policy that satisfies this condition is strongly stable. 

\subsection{Stability of \name}

We begin by formally stating our considered concept of stability. 
\begin{definition}[Strong stability] We say that the system is strongly stable {\em iff} there exists a constant $K \ge 0$ such that 
$$\limsup_{T \to \infty} \frac{1}{T} \sum_{t=0}^{T-1} \sum_{i=1}^n \E \Big[Q_i(t)\Big] \le K.$$
That is, the system is strongly stable when the expected time averaged sum of queue lengths admits a constant upper bound.   
\end{definition}
Strong stability is a strong form of stability that implies finite
average backlog and (by Little's theorem) finite average delay. Furthermore, under mild conditions, it implies other commonly considered forms of stability, such as steady state
stability, rate stability, mean rate stability and more (see \cite{neely2012stability}). Note that strong stability has been widely used in queuing systems (see \cite{georgiadis2006resource} and references therein) whose state does not necessarily admit an irreducible and aperiodic Markov chain representation (therefore positive-recurrence may not be considered).

\begin{theorem}\label{thm:1}
    Assume that the system uses \name{} and Assumption \ref{asm_2} holds.
    Then the system is strongly stable.
\end{theorem}
 
\begin{proof}
A server can work on a job immediately upon its arrival. Therefore, the queue dynamics at server $i$ are given by
\begin{equation}\label{eq:workload_dynamics}
  Q_i(t+1) = [Q_i(t) + a_i(t) - s_i(t)]^+,  
\end{equation}
 where $[\cdot]^+ \equiv \max\set{\cdot,0}$. Squaring both sides of \eqref{eq:workload_dynamics} yields 
\begin{equation}\label{eq:square_w}
\begin{split}
    &\bp{Q_i(t+1)}^2 \le \bp{Q_i(t)}^2  + \bp{a_i(t)}^2 + \bp{s_i(t)}^2 \cr 
    & + 2a_i(t)Q_i(t) - 2s_i(t)Q_i(t) - 2s_i(t)a_i(t).   
\end{split}
\end{equation} 
Rearranging \eqref{eq:square_w} and omitting the last term yields 
\begin{equation}
\begin{split}
    &\bp{Q_i(t+1)}^2 - \bp{Q_i(t)}^2 \le \cr 
    &\bp{a_i(t)}^2 + \bp{s_i(t)}^2 - 2Q_i(t)\bp{s_i(t)-a_i(t)}.    
\end{split}
\end{equation}
Summing over the servers yields
\begin{equation}\label{eq:sum_over_servers}
\begin{split}
  &\sum_{i=1}^n \bp{Q_i(t+1)}^2 - \sum_{i=1}^n \bp{Q_i(t)}^2 \le \cr  
  &B(t) - 2\sum_{i=1}^n Q_i(t)\bp{s_i(t)-a_i(t)},
\end{split}
\end{equation}
where
\begin{equation}
  B(t) = \sum_{i=1}^n \bp{a_i(t)}^2 + \sum_{i=1}^n \bp{s_i(t)}^2.
\end{equation}
We would like to proceed by taking the expectation of \eqref{eq:sum_over_servers}. To that end, we need to analyze the term $$\sum_{i=1}^n Q_i(t)\bp{s_i(t)-a_i(t)},$$ since $\set{Q_i(t)}_{i=1}^n$ and $\set{a_i(t)}_{i=1}^n$ are dependent.

We shall conduct the following plan. We will introduce two additional policies into our analysis: (1)~$JSQ$ and (2)~Weighted-Random ($WR$). Roughly speaking, we will show that the routing decision that is taken by our policy at each dispatcher and each time slot $t$ is sufficiently similar to the decision that would have been made by $JSQ$ \emph{given the same system state}, which, in turn, is no worse than the decision that $WR$ would make at that time slot. Since in $WR$ the routing decisions taken at time slot $t$ do not depend on the system state at time slot $t$, we will obtain the desired independence, which allows us to continue with the analysis. 

We start by introducing the corresponding $JSQ$ and $WR$ notations. Let $$a_i^{JSQ}(t) = \sum_{j=1}^m a_i^{j,JSQ}(t)$$ be the number of jobs that will be routed to server $i$ at time slot $t$ when using $JSQ$ \emph{at time slot $t$}. That is, each dispatcher forwards its incoming jobs to the server with the shortest queue (ties are broken randomly). Formally, let $i^* \in argmin_i\set{Q_i(t)}$, then $\forall j \in M$
\begin{equation}
a_i^{j,JSQ}(t) = 
\begin{cases}
a^j(t), & i = i^* \\ 
0,      & i \neq i^*.
\end{cases}
\end{equation}

Let $$a_i^{WR}(t) = \sum_{j=1}^m a_i^{j,WR}(t)$$ be the number of jobs that will be routed to server $i$ at time slot $t$ when using $WR$ \emph{at time slot $t$}. That is, each dispatcher forwards its incoming jobs to a single randomly-chosen server, where the probability of choosing server $i$ is $\frac{\mu_i^{(1)}}{\sum_{i=1}^n \mu_i^{(1)}}$. Formally, $\forall j \in M$, $i=i^*$ with probability $\frac{\mu_i^{(1)}}{\sum_{i=1}^n \mu_i^{(1)}}$ and
\begin{equation}
a_i^{j,WR}(t) = 
\begin{cases}
a^j(t), & i = i^* \\ 
0,      & i \neq i^*
\end{cases}
\end{equation}

With these notations at hand, we continue our analysis by adding and subtracting the term $2\sum_{i=1}^n  a_i^{JSQ}(t) Q_i(t)$ from the right hand side of \eqref{eq:sum_over_servers}. This yields 
\begin{equation}\label{eq:sum_over_servers_add}
\begin{split}
  & \sum_{i=1}^n \bp{Q_i(t+1)}^2 - \sum_{i=1}^n \bp{Q_i(t)}^2 \le \cr
  & B(t) - 2\sum_{i=1}^n Q_i(t)\bp{s_i(t)-a_i^{JSQ}(t)} + \cr
  & 2\sum_{i=1}^n Q_i(t)\bp{a_i(t)-a_i^{JSQ}(t)}.
\end{split}
\end{equation}

We would like to take the expectation of \eqref{eq:sum_over_servers_add}. However, as mentioned, since the actual queue lengths and the local views of the dispatchers and the routing decisions that are made both by our policy and $JSQ$ are dependent, we shall rely on the $WR$ policy and the expected distance of the local views from the actual queue lengths to evaluate the expected values. To that end, we introduce the following lemmas. 

\begin{lemma}\label{lem:JSQ_wr}
For all time slots $t$, it holds that 
$$\sum_{i=1}^n a_i^{JSQ}(t) Q_i(t) \le \sum_{i=1}^n a_i^{WR}(t) Q_i(t).$$ 
\end{lemma}
\begin{proof}
    See Appendix \ref{app:lem:JSQ_wr}.
\end{proof}
\begin{lemma}\label{lem:diff_w_JSQ}
For all servers $i \in N$ and all time slots $t$, it holds that 
\begin{equation*}
\begin{split}
    \sum_{i=1}^n Q_i(t)\bp{a_i(t)-a_i^{JSQ}(t)} \le \sum_{i=1}^n \sum_{j=1}^m a(t)\Big|Q_i(t)-\tilde{Q}_i^j(t)\Big|.
\end{split}
\end{equation*}
\end{lemma}
\begin{proof}
    See Appendix \ref{app:lem:diff_w_JSQ}.
\end{proof}

Applying Lemmas \ref{lem:JSQ_wr} and \ref{lem:diff_w_JSQ} to \eqref{eq:sum_over_servers_add} yields
\begin{equation}\label{eq:sum_over_servers_add_3}
\begin{split}
  &\sum_{i=1}^n \bp{Q_i(t+1)}^2 - \sum_{i=1}^n \bp{Q_i(t)}^2 \le \cr
  & B(t) - 2\sum_{i=1}^n Q_i(t)\bp{s_i(t)-a_i^{WR}(t)} \cr 
  & + 2\sum_{i=1}^n \sum_{j=1}^m a(t) \Big|Q_i(t)-\tilde{Q}_i^j(t)\Big|.
\end{split}
\end{equation}
Taking the expectation of \eqref{eq:sum_over_servers_add_3} yields 
\begin{equation}\label{eq:ex_drift}
\begin{split}
   & \E \Big[\sum_{i=1}^n \bp{Q_i(t+1)}^2\Big] -  \E \Big[\sum_{i=1}^n \bp{Q_i(t)}^2\Big] \le \cr  
   & \E \Big[B(t)\Big] - 2\E \Big[ \sum_{i=1}^n Q_i(t)\bp{s_i(t)-a_i^{WR}(t)}\Big] \cr 
   & + 2\E \Big[ \sum_{i=1}^n \sum_{j=1}^m a(t) \Big|Q_i(t)-\tilde{Q}_i^j(t)\Big|\Big].   
\end{split}
\end{equation}
We observe that both $a(t)$ (according to \eqref{eq:arrival_1}) and $\set{a_i^{WR}(t)}_{i=1}^n$ (according to the definition of the $WR$ policy) are independent of $\set{Q_i(t)}_{i=1}^n$ and $\set{\tilde{Q}_i^j(t) \given{(i,j)\in N \times M}}$. {\new{Specifically, by~\eqref{eq:arrival_1}, $a(t)$ is independent of any events prior to time $t$, including the number of accumulated jobs in the system by time $t$}}. Applying this observation to \eqref{eq:ex_drift} and using the linearity of expectation yields
\begin{equation}\label{eq:ex_drift_2}
\begin{split}
   & \E \Big[\sum_{i=1}^n \bp{Q_i(t+1)}^2\Big] -  \E \Big[\sum_{i=1}^n \bp{Q_i(t)}^2\Big] \le \cr  
   & \E \Big[B(t)\Big] - 2 \sum_{i=1}^n \E \Big[Q_i(t)\Big]\E \Big[\bp{s_i(t)-a_i^{WR}(t)}\Big] \cr 
   & + 2 \sum_{i=1}^n \sum_{j=1}^m \E \Big[a(t)\Big] \E \Big[\Big|Q_i(t)-\tilde{Q}_i^j(t)\Big|\Big].   
\end{split}
\end{equation}
Next, since for any non-negative $\set{x_1,x_2,\ldots,x_n}$ such that $$x=x_1+x_2+\ldots+x_n$$
it always holds that $x^2 \ge \sum_{i=1}^n x_i^2$, using \eqref{eq:service_1}-\eqref{eq:service_3}, the linearity of expectation and \eqref{eq:arrival_1}-\eqref{eq:arrival_3}, we obtain
\begin{equation}\label{eq:exp_b}
\begin{split}
   &\E \Big[B(t)\Big] =  \E \Big[\sum_{i=1}^n \bp{a_i(t)}^2\Big] + \E \Big[\sum_{i=1}^n \bp{s_i(t)}^2\Big] \le \cr
   & \E \Big[\bp{a(t)}^2\Big] + \sum_{i=1}^n \E \Big[\bp{s_i(t)}^2\Big] = \lambda^{(2)} + \sum_{i=1}^n \mu_i^{(2)}.
\end{split}
\end{equation}
Additionally, using \eqref{eq:arrival_1}, \eqref{eq:arrival_2} and \eqref{eq:disp_knowledge} yields
\begin{equation}\label{eq:exp_assump}
\begin{split}
&\sum_{i=1}^n \sum_{j=1}^m \E \Big[a(t)\Big] \E \Big[\Big|Q_i(t)-\tilde{Q}_i^j(t)\Big|\Big] \le \cr
&\sum_{i=1}^n \sum_{j=1}^m \lambda^{(1)} C = m n \lambda^{(1)} C.
\end{split}
\end{equation}
Finally, since the decisions taken by the $WR$ policy are independent of the system state, we can introduce the following lemma.
\begin{lemma}\label{lem:wr_stable}
For all $i \in N$ and $t$ it holds that 
\begin{equation}\label{eq:wr_stable}
    \E \Big[s_i(t)-a_i^{WR}(t)\Big] = \frac{\epsilon\mu_i^{(1)}}{\sum_{i=1}^n \mu_i^{(1)}}.
\end{equation}
\end{lemma}
\begin{proof}
    See Appendix \ref{app:lem:wr_stable}.
\end{proof}

Using \eqref{eq:exp_b}, \eqref{eq:exp_assump} and Lemma \ref{lem:wr_stable} in \eqref{eq:ex_drift_2} yields 
\begin{equation}\label{eq:ex_drift_3}
\begin{split}
   & \E \Big[\sum_{i=1}^n \bp{Q_i(t+1)}^2\Big] -  \E \Big[\sum_{i=1}^n \bp{Q_i(t)}^2\Big] \le \cr 
   & \lambda^{(2)} + \sum_{i=1}^n \mu_i^{(2)} + 2 m n \lambda^{(1)} C \cr 
   & - 2 \sum_{i=1}^n \frac{\epsilon\mu_i^{(1)}}{\sum_{i=1}^n \mu_i^{(1)}} \E \Big[Q_i(t)\Big].   
\end{split}
\end{equation}
For ease of exposition, denote the constants 
\begin{equation}\label{eq:const_1}
D = \lambda^{(2)} + \sum_{i=1}^n \mu_i^{(2)} + 2 m n \lambda^{(1)} C,
\end{equation}
and 
\new{\begin{equation}\label{eq:const_2}
\delta = \frac{\epsilon}{\sum_{i=1}^n \mu_i^{(1)}}.
\end{equation}}

Rearranging \eqref{eq:ex_drift_3} and using \eqref{eq:const_1} and \eqref{eq:const_2} yields
\new{\begin{equation}\label{eq:ex_drift_4}
\begin{split}
    &  2\delta \sum_{i=1}^n \mu_i^{(1)} \cdot \E \Big[Q_i(t)\Big] \le \cr
    & D + \Bigg( \E \Big[\sum_{i=1}^n \bp{Q_i(t)}^2\Big] - \E \Big[\sum_{i=1}^n \bp{Q_i(t+1)}^2\Big]\Bigg).
\end{split}
\end{equation}}
Summing \eqref{eq:ex_drift_4} over time slots $[0,1,\ldots,T-1]$, noticing the telescopic series at the right hand side of the inequality and dividing by $2\delta T$ yields
\new{\begin{equation}\label{eq:ex_drift_5}
\begin{split}
    &  \frac{1}{T} \sum_{t=0}^{T-1} \sum_{i=1}^n \mu_i^{(1)} \cdot  \E \Big[Q_i(t)\Big] \le \frac{D}{2\delta} + \cr
    &  \frac{1}{2\delta T} \Bigg( \E \Big[\sum_{i=1}^n \bp{Q_i(0)}^2\Big] - \E \Big[\sum_{i=1}^n \bp{Q_i(T)}^2\Big]\Bigg).
\end{split}
\end{equation}}
Taking limits of \eqref{eq:ex_drift_5} and making the standard assumption that the system starts its operation with finite queue lengths, \ie $$\E \Big[\sum_{i=1}^n \bp{Q_i(0)}^2\Big] < \infty$$ yields 
\new{\begin{equation}\label{eq:ex_drift_6}
\begin{split}
    &  \limsup_{T \to \infty} \frac{1}{T} \sum_{t=0}^{T-1} \sum_{i=1}^n \mu_i^{(1)} \cdot \E \Big[Q_i(t)\Big] \le \frac{D}{2\delta}.
\end{split}
\end{equation}}
\new{Now, dividing both sides of \eqref{eq:ex_drift_6} by $\min_i \{ \mu_i^{(1)} \}$} yields
\new{\begin{equation}\label{eq:ex_drift_7}
\begin{split}
    &  \limsup_{T \to \infty} \frac{1}{T} \sum_{t=0}^{T-1} \sum_{i=1}^n \frac{\mu_i^{(1)} \cdot \E \Big[Q_i(t)\Big]}{\min_i \{ \mu_i^{(1)}\}} \le \frac{D}{2\delta \cdot {\min_i \{ \mu_i^{(1)}\}}}.
\end{split}
\end{equation}}
\new{Finally, using the fact that $$\frac{\mu_i^{(1)} \cdot \E \Big[Q_i(t)\Big]}{\min_i \{ \mu_i^{(1)}\}} \ge \E \Big[Q_i(t)\Big] \quad \forall i$$ in \eqref{eq:ex_drift_7}} we obtain 
\new{\begin{equation}\label{eq:ex_drift_8}
\begin{split}
    &  \limsup_{T \to \infty} \frac{1}{T} \sum_{t=0}^{T-1} \sum_{i=1}^n \E \Big[Q_i(t)\Big] \le \frac{D}{2\delta \cdot {\min_i \{ \mu_i^{(1)}\}}}.
\end{split}
\end{equation}}
This implies strong stability and concludes the proof.
\end{proof}
\new{Note that the result in \eqref{eq:ex_drift_6} yields a slightly stronger bound than the one achieved in \eqref{eq:ex_drift_8}. This is because in \eqref{eq:ex_drift_6} we have a bound on the {service-rate-weighted} time-averaged expected sum of queue lengths that is not sensitive to the skew among server service rates. That is, even for an arbitrary large skew (\ie when $\frac{\min_i \{ \mu_i^{(1)}\}}{\max_i \{ \mu_i^{(1)}\}}$ is arbitrary small), the bound does not grow to infinity. Nevertheless, we provide the result in \eqref{eq:ex_drift_7} to obtain the standard form of strong stability (without service-rate normalization of the queue sizes).} 

\section{Simplified stability conditions}

As mentioned, in order to establish that a system that uses an \name{} policy is strongly stable, it is sufficient to show that Assumption \ref{asm_2} holds. Generally, it may be challenging to establish this. To that end, we now develop simplified sufficient conditions. As we later demonstrate, these simplified conditions capture broad families of communication techniques between the servers and the dispatchers, and allow for the design of stable policies with appealing performance and extremely low communication budgets.

Throughout the proofs, we assume that our sufficient condition always holds when the system starts its operation, namely that  there exists a constant $C_0 \ge 0$ such that
\begin{equation}\label{eq:time_0_approx}
    \E \Big[\big|Q_i(0) - \tilde{Q}_i^j(0)\big| \Big] \le C_0 \,\, \forall \, (i,j) \in N \times M.
\end{equation}
Also, we denote $\mathbbm{1}_i^j(t)$ as an indicator function that obtains the value $1$ iff server $i$ updates dispatcher $j$ (via the push-based sampling or the pull-based update message from the server) with its actual queue length at the end of time slot $t$ (after arrivals and departures at time slot $t$).

\subsection{Stochastic updates}

We now prove that in \name{}, it is sufficient for the system to be strongly stable if for any server $i$, any dispatcher $j$, and any current local state error at dispatcher $j$ for server $i$, there is a strictly positive probability that dispatcher $j$ receives an update from server $i$. Intuitively, it means that the dispatcher may be rarely updated, but expected times between updates are still finite, and therefore errors do not grow unbounded in expectation.

\begin{theorem}\label{thm:simple_cond}
Assume that there exists $\bar{\epsilon}>0$ such that
\begin{equation}\label{eq:simp_asm}
    \E \Big[\mathbbm{1}_i^j(t) \,\, \Big| \,\, \big|Q_i(t) - \tilde{Q}_i^j(t)\big| \Big] > \bar{\epsilon} \quad \forall (i,j,t) \in N \times M \times \mathbb{N}.
\end{equation}
Then, Assumption \ref{asm_2} holds and the system is strongly stable. 
\end{theorem}

\begin{proof}
Fix server $i$ and dispatcher $j$. Denote 
$$Z(t) = \big|Q_i(t) - \tilde{Q}_i^j(t)\big|.$$ 
Now, for all $t$ it holds that
\begin{equation}\label{eq:z_t}
\begin{split}
&Z(t+1) \le \big(1-\mathbbm{1}_i^j(t)\big) \cdot \Big(Z(t) + a_i(t) + s_i(t)\Big) \le \cr 
&\big(1-\mathbbm{1}_i^j(t)\big) \cdot Z(t) + a(t) + s_i(t).
\end{split}
\end{equation}
Taking expectation of \eqref{eq:z_t} yields
\begin{equation}\label{eq:z_t_exp}
\E [Z(t+1)] \le \E[(1-\mathbbm{1}_i^j(t)\big) \cdot Z(t)] + \lambda^{(1)} + \mu_i^{(1)}.
\end{equation}
Next, using the law of total expectation
\begin{equation}\label{eq:z_t_exp_2}
\begin{split}
  & \E[(1-\mathbbm{1}_i^j(t)\big) \cdot Z(t)] = \cr 
  & \E\Big[\E\big[(1-\mathbbm{1}_i^j(t)\big) \cdot Z(t) \,\big|\, Z(t)\big]\Big] =\cr
  & \E \Big[ Z(t) \cdot \E\big[(1-\mathbbm{1}_i^j(t)\big) \,\big|\, Z(t)\big]\Big] \le (1-\bar{\epsilon})\E[Z(t)],
\end{split}
\end{equation}
where the last inequality follows from the linearity of expectation and \eqref{eq:simp_asm}. Now, using \eqref{eq:z_t_exp_2} in \eqref{eq:z_t_exp} yields
\begin{equation}\label{eq:z_t_exp_3}
\begin{split}
&\E [Z(t+1)] \le  (1-\bar{\epsilon})\E[Z(t)] + \lambda^{(1)} + \mu_i^{(1)}. 
\end{split}
\end{equation}
We now introduce the following lemma.
\begin{lemma}\label{lem:gap_reccurence}
Fix $\epsilon \in (0,1]$, $C_1 \ge 0$ and $C_2 \ge 0$. Consider the recurrence 
$$T(n+1) \le (1-\epsilon) \cdot T(n) + C_1,$$ 
with the initial condition 
$$T(0) \le C_2.$$ 
Then, 
$$T(n) \le \max \set{\frac{C_1}{\epsilon}, C_2} \,\, \forall \, n.$$
\end{lemma}
\begin{proof}
    See Appendix \ref{app:lem:gap_reccurence}.
\end{proof}
Finally, using Lemma \ref{lem:gap_reccurence} in \eqref{eq:z_t_exp_3} yields
$$\E [Z(t)] \le \max \set{\frac{\lambda^{(1)} +\mu_i^{(1)}}{\bar{\epsilon}}, C_0} \,\, \forall \, t.$$
This concludes the proof.
\end{proof}

\subsection{Deterministic updates}

We proceed to establish that any \name{} policy, in which each local view entry is updated at least once every fixed number $C_{up}$ of time slots, is strongly stable.    

\begin{theorem}\label{thm:simple_cond_2}
Assume that each local entry is updated at least once every $C_{up} \in \mathbbm{N}$ time slots. Then Assumption \ref{asm_2} holds and the system is strongly stable.
\end{theorem}

\begin{proof}
Fix server $i$ and dispatcher $j$.
Then,  
\begin{equation}
   |Q_i(t+C_{up}) - \tilde{Q}_i^j(t+C_{up})| \le \sum_{\tau=t}^{t+C_{up}-1} (a(\tau)+s_i(\tau)) 
\end{equation}
This is because the last update of this entry happened at most $C_{up}$ time slots ago. Now, by taking the expectation, we obtain
\begin{equation}
\begin{split}
       &\E\big[|Q_i(t+C_{up}) - \tilde{Q}_i^j(t+C_{up})|\big] \le \cr 
       &\E\sbrac{\sum_{\tau=t}^{t+C_{up}-1} (a(\tau)+s_i(\tau))} = C_{up}(\lambda^{(1)}+\mu_i^{(1)})  
\end{split}
\end{equation}
On the other hand, for any $t < C_{up}$ it holds that
\begin{equation}
\begin{split}
       \E\big[|Q_i(t) - \tilde{Q}_i^j(t)|\big] & \le C_0 + \E\big[\sum_{\tau=0}^{C_{up}-1} (a(\tau)+s_i(\tau))\big] \cr 
       & \le C_0 + C_{up}(\lambda^{(1)}+\mu_i^{(1)})
\end{split}
\end{equation}
This concludes the proof.
\end{proof}


\section{Example \name{} policies}

Since \name{} is not restricted to work with either pull- or push-based communications, in this section we provide examples for both. In a push-based policy, the dispatchers sample the servers for their queue lengths, whereas in a pull-based policy the servers update the dispatchers with their queue lengths. While empirically we will see that the pull-based approach can provide better performance in many scenarios, it may also incur additional implementation overhead, as it requires the servers to actively update the dispatchers given some state conditions, rather than passively answer sample queries. Therefore, we consider both the push and pull frameworks.

\subsection{Push-based \name{} example}
The power--of-choice $JSQ(d)$ policy forms a popular low-communication push-based load balancing approach, but it is not stable in heterogeneous systems, even for a single dispatcher. Instead, we will now analyze a push-based \name{} policy that uses exactly the same communication pattern between the servers and the dispatchers. It essentially extends the policy of \cite{jonatha2018power}, which considers a single dispatcher and homogeneous servers, to multiple dispatchers with heterogeneous servers.

In this policy, which we call \name{}-$Sample(d)$ and describe in Algorithm \ref{alg:lib_poc}, each dispatcher holds a local, possibly outdated, array of the server queue lengths, and sends jobs to the minimum one among them. The array entries are updated as follows: (1) when a dispatcher sends jobs to a server, these jobs are added to the respective local approximation; (2)~at each time slot, if new jobs arrive, the dispatcher randomly samples $d$ distinct queues and uses this information only to update the respective $d$ distinct entries in its local array to their actual value.

The simplicity of \name{}-$Sample(d)$ may be surprising. For instance, there is no attempt to guess or estimate how the other dispatchers send traffic or how the queue drains to get a better estimate, \ie our estimate is based only on the jobs that the specific dispatcher sends and the last time it sampled a queue. We also do not take the age of the information into account. 

Furthermore, as we find below,  the stability proof of \name{}-$Sample(d)$ only relies on the sample messages and not on the job increments. We empirically find that these increments help improve the estimation quality and therefore the performance.

\begin{algorithm}[!t]
\small
\caption{\name{}-$Sample(d)$ (push-based comm.)}
\label{alg:generator}
\SetKwProg{generate}{}{}{end}
\uline{Code for dispatcher $j \in M$:}\\
\smallskip
\generate{\textbf{Route jobs and update local state:}}{
    \ForEach{time slot $t$}{
        Forward jobs to server $i^* \in argmin_i \set{\tilde{Q}_i^j(t)}$\;
        Update $\tilde{Q}_{i^*}^j(t) \gets \tilde{Q}_{i^*}^j(t) + a^j(t)$\;
    }
}
\smallskip
\generate{\textbf{Sample servers and update local state:}}{
    \ForEach{time slot $t$}{
        \uIf{new jobs arrive at time slot $t$}{
            Uniformly at random pick distinct $i_1, \ldots,i_d \in N$\;
            For each $i \in \{i_1,\dots,i_d\}$ update $\tilde{Q}_{i}^j(t) \gets Q_{i}(t)$\;
        }
    }
}
\label{alg:lib_poc}
\end{algorithm}

We proceed to establish that using \name{}-$Sample(d)$ at each dispatcher results in strong stability in multi-dispatcher heterogeneous systems. Interestingly, this result holds even for $d=1$.

\begin{proposition}\label{prop:1}
Assume that the system uses \name{}-$Sample(d)$. Then, it is strongly stable. 
\end{proposition}
\begin{proof}
Fix dispatcher $j$ and server $i$. Consider time slot $t$. By \eqref{eq:arrival_4}, with probability of at least $\epsilon_0$, dispatcher $j$ samples $d$ out of $n$ servers uniformly at random disregarding the system state at time slot $t$. Therefore, we obtain 
$$\E \Big[\mathbbm{1}_i^j(t) \,\, \Big| \,\, \big|Q_i(t) - \tilde{Q}_i^j(t)\big| \Big] \ge \frac{\epsilon_0 \cdot d}{n}.$$
This respects the simplified probabilistic sufficiency condition and thus concludes the proof. 
\end{proof}

\subsection{Pull-based \name{} example}

$JIQ$ is a popular, recently proposed, low-communication pull-based load balancing policy. It offers a low communication overhead that is upper-bounded by a single message per job~\cite{lu2011join}. However, as mentioned, for heterogeneous systems, $JIQ$ is not stable even for a single dispatcher. 

We now propose a different pull-based \name{} policy that conforms with the same communication upper bound, namely a single message per job, and leverages the important idleness signals from the servers. It essentially follows similar lines to the policy presented in \cite{van2019hyper}, which considers a single dispatcher and homogeneous servers, and extends it to multiple dispatchers with heterogeneous servers. 

Specifically, each server, upon the completion of one or several jobs at the end of a time slot, sends its queue length to a dispatcher, which is chosen uniformly at random, using the following rule: (1) if the server becomes idle, then the message is sent with probability 1; (2) otherwise, the message is sent with probability $0 < p \le 1$ where $p$ is a fixed, arbitrary small, parameter. 

Algorithm \ref{alg:lib_jiq} (termed \name{}-$Update(p)$) depicts the actions taken by each dispatcher at each time slot.

The intuition behind this approach is to always leverage the idleness signals in order to avoid immediate starvation as done by $JIQ$; yet, in contrast to $JIQ$, even when no servers are idle, we want to make sure that the local views are not too far from the real queue lengths, which provides significant advantage at high loads. 

\begin{algorithm}[!t]
\small
\caption{\name{}-$Update(p)$ (pull-based comm.)}
\label{alg:generator2}
\SetKwProg{generate}{}{}{end}
\uline{Code for dispatcher $j \in M$:}\\
\smallskip
\generate{\textbf{Route jobs and update local state:}}{
    \ForEach{time slot $t$}{
        Forward jobs to server $i^* \in argmin_i \set{\tilde{Q}_i^j(t)}$\;
        Update $\tilde{Q}_{i^*}^j(t) \gets \tilde{Q}_{i^*}^j(t) + a^j(t)$\;
    }
}
\smallskip
\generate{\textbf{Update local state:}}{
    \ForEach{arrived message $\langle i , q \rangle$ at time slot $t$}{
        Update $\tilde{Q}_{i}^j(t) \gets q$\;
    }
}
\uline{Code for server $i \in N$:}\\
\smallskip
\generate{\textbf{Send update message:}}{
    \ForEach{time slot $t$}{
        \uIf{completed jobs at time slot $t$}{
            Uniformly at random pick  $j \in M$\;
            \uIf{idle}{
                Send $\langle i , Q_i(t) \rangle$ to dispatcher $j$\;
            }
            \Else{
                Send $\langle i , Q_i(t) \rangle$ to dispatcher $j$ $w.p.$ $p$\;
            }
        }
    }
}
\label{alg:lib_jiq}
\end{algorithm}

We now formally prove that using \name{}-$Update(p)$ results in strong stability in multi-dispatcher heterogeneous systems. Interestingly, this result holds for any $p>0$.

\begin{proposition}\label{prop:2}
Assume that the system uses \name{}-$Update(p)$. Then, it is strongly stable.
\end{proposition}

\begin{proof}
We prove that \eqref{eq:simp_asm} holds. Fix dispatcher $j$, server $i$ and time slot $t$. We examine two possible events at the beginning of time slot $t$: (1) $Q_i(t)=0$ and (2) $Q_i(t)>0$. 

\TT{(1)} Since $Q_i(t)=0$, the server updated at least one dispatcher in a previous time slot, \ie for at least one dispatcher $j^*$ we have that $\tilde{Q}_{i}^{j^*}(t)=0$. This must hold since there is a dispatcher that received the update message after this queue got empty (that is, when a server becomes idle, a message is sent $w.p.$ 1). Now consider the event $A_1 = \set{a_i^{j^*}(t)>0 \cap s_i(t)>0}$. Since the tie breaking rule is random, by \eqref{eq:arrival_1}, \eqref{eq:arrival_2}, \eqref{eq:arrival_4}, \eqref{eq:service_1} and \eqref{eq:service_2}, there exists $\bar{\epsilon}_i>0$ such that $\mathbb{P}(A_1)>\bar{\epsilon}_i$. Since $a(t)$ and $s_i(t)$ are not dependent on any system information at the beginning of time slot $t$ we obtain
\begin{equation}\label{eq:jiq_lib_proof_1}
\begin{split}
    \E \Big[\mathbbm{1}_i^j(t) \,\, \Big| \,\, \big|Q_i(t) - \tilde{Q}_i^j(t)\big|, & Q_i(t) = 0 \Big] \ge \cr
    & p \cdot \mathbb{P}(A_1)> p \cdot \bar{\epsilon}_i. 
\end{split}
\end{equation}

\TT{(2)} Since $Q_i(t)>0$, there is a strictly positive probability that a job would be completed at this time slot. That is, since $s_i(t)$ is not dependent on any system information at the beginning of time slot $t$ we obtain
\begin{equation}\label{eq:jiq_lib_proof_2}
\begin{split}
    \E \Big[\mathbbm{1}_i^j(t) \,\, \Big| \,\, \big|Q_i(t) - \tilde{Q}_i^j(t)\big|, & Q_i(t) > 0 \Big] \ge \cr
    & p \cdot P(s_i(t)>0) > p \cdot \bar{\epsilon}_i. 
\end{split}
\end{equation}

Finally, since $\E \Big[\mathbbm{1}_i^j(t) \,\, \Big| \,\, \big|Q_i(t) - \tilde{Q}_i^j(t)\big| \Big]$ is a convex combination of the left hand sides of \eqref{eq:jiq_lib_proof_1} and \eqref{eq:jiq_lib_proof_2} we obtain that
$$\E \Big[\mathbbm{1}_i^j(t) \,\, \Big| \,\, \big|Q_i(t) - \tilde{Q}_i^j(t)\big| \Big]>p \cdot \bar{\epsilon}_i.$$ 
Now fix $\bar{\epsilon} = \min_i \set{\bar{\epsilon}_i}$. We have that
$$\E \Big[\mathbbm{1}_i^j(t) \,\, \Big| \,\, \big|Q_i(t) - \tilde{Q}_i^j(t)\big| \Big]>p \cdot \bar{\epsilon} \quad \forall (i,j,t) \in N \times M \times \mathbb{N}.$$
Again, the probabilistic sufficiency condition holds and this concludes the proof. 
\end{proof}

\subsection{\name{} with smart servers}

We next propose a more advanced \name{} variant that relies on smart servers.
Namely, when the system uses a pull-based communication type, the servers update the dispatchers. As a result, the servers know the dispatcher states and how bad their local views are. This is because a server knows what was the last update message it sent to a dispatcher and how many jobs it has received from it since then.

In terms of additional resource requirements, an array of size $m$ is sufficient in order to keep this information at a server. Namely, each server $i$ holds a single entry for each dispatcher $j$. This entry is updated when server $i$ sends an update message to dispatcher $j$ or when new jobs from dispatcher $j$ arrive to server $i$. 

Specifically, we propose \name{}-$Smart(f,A)$. This \name{} variant is defined via two parameters, as follows. 
\begin{enumerate}
\item A function $f$ used by each server. Namely, $f$ determines the probability by which a server, at each time slot in which it completes at least one job, sends an update message. 
\item An algorithm $A$ used by each server. Namely, $A$ determines which dispatcher will be updated if a message is sent. 
\end{enumerate}
Both $f$ and $A$ may depend on the server's identity and its state. The pseudo-code for \name{}-$Smart(f,A)$ server update messages is presented by Algorithm \ref{alg:lib_smart}. The remainder of the dispatcher's code is identical to Algorithm \ref{alg:lib_jiq}. 

\begin{algorithm}[!t]
\small
\caption{\name{}-$Smart(f,A)$ (smart servers)}
\label{alg:generator3}
\SetKwProg{generate}{}{}{end}
\uline{Code for server $i \in N$:}\\
\smallskip
\generate{\textbf{Send update message:}}{
    \ForEach{time slot $t$}{
        \uIf{completed jobs at time slot $t$}{
         with probability $f$:\\ 
         \quad\quad send $\langle i , Q_i(t) \rangle$ to dispatcher dictated by $A$\;
        }
    }
}
\label{alg:lib_smart}
\end{algorithm}

For concreteness, we next choose specific $f$ and $A$ and prove that they result in a stable \name{} policy. Note that we choose $f$ and $A$ heuristically and not via optimization. We leave further investigation regarding smart servers to future work.

\smallskip
\TT{Load dependent updates.} For server $i$ and dispatcher $j$ denote
$$Z_i(t) = \max_j \{ \big|Q_i(t) - \tilde{Q}_i^j(t)\big| \}.$$
Our heuristic choice for $f$ is given by 
$$f(i,p,t) = 
\begin{cases}
    p \quad Z_i(t) < Q_i(t) \\
    1 \quad Z_i(t) \ge Q_i(t),
\end{cases}
$$
where $p \in (0,1]$ is a fixed constant. In other words, at each time slot in which server $i$ completes at least one job, it decides whether to update a dispatcher: (1) the server only updates a dispatcher w.p. (with probability) $p$ whenever all local-state errors at dispatchers are relatively small, \ie strictly smaller than the server queue state; (2) Else, it sends an update w.p. 1. 

For instance, if at the beginning of a time slot, the server queue size is 10 and a dispatcher thinks that it is 0 or 20 (\ie its local error is not small), then the server necessarily sends an update whenever it completes a job at that time slot. Another example is if the server \emph{becomes idle}, and therefore the error cannot be smaller than the queue size of 0, hence the server also sends an update w.p. 1. 

\smallskip
\TT{Worst approximation first (\mbox{WAF}).} We set $A$ such that a server always updates the dispatcher that has the worst local view when the message is sent. That is, if server $i$ sends an update message at time slot $t$, it is sent to dispatcher ${j^* \in argmax_j\{\big|Q_i(t) - \tilde{Q}_i^j(t)\big|\}}$. Ties are broken arbitrarily. 

For ease of exposition, we abuse notation and denote the function $f^*(p) = \{f(i,p,t)\}_{i \in N, t \in \mathbbm{N}}$. We now turn to prove that using \name{}-$Smart(f^*(p),\mbox{WAF})$ results in strong stability for any $p>0$.

\begin{proposition}\label{prop:3}
Assume that the system is sub-critical and uses \name{}-$Smart(f^*(p),\mbox{WAF})$. Then, the system is strongly stable.
\end{proposition}

\begin{proof} 

We already saw in the proof of Theorem \ref{thm:simple_cond} that, if 
\begin{equation}
    \E \Big[\mathbbm{1}_i^j(t) \,\, \Big| \,\, \big|Q_i(t) - \tilde{Q}_i^j(t)\big| \Big] > \bar{\epsilon} \quad \forall (i,j,t) \in N \times M \times \mathbb{N}.
\end{equation}
then Assumption \ref{asm_2} holds and we have strong stability. However, when using \name{}-$Smart(f^*(p),\mbox{WAF})$, such a condition may fail to hold. Indeed, when a server sends an update message, it is sent to the dispatcher that holds the worst local view at that time slot, thus leaving other dispatchers, with possible better local views, with a probability of $0$ for an update. Therefore, we simply show that updating the worst local view is at least as good as updating any dispatcher that is chosen randomly. We do so by examining the \emph{sum of errors} over the dispatcher for a specific server $i$. Fix server $i$ and denote
$$Z_j(t) = \big|Q_i(t) - \tilde{Q}_i^j(t)\big|.$$ 
Now, for all $t$ it, holds that
\begin{equation}\label{eq:z_t_smart}
\begin{split}
&\sum_jZ_j(t+1) \le \cr 
&\sum_j\big(1-\mathbbm{1}_i^{j}(t)\big) \cdot \Big(Z_j(t) + a_i(t) + s_i(t)\Big) \le \cr   
&\sum_j\Big(\big(1-\mathbbm{1}_i^{j}(t)\big) \cdot Z_j(t) + a(t) + s_i(t)\Big) \le \cr
&\sum_j\Big(\big(1-\mathbbm{1}_i^{*,j}(t)\big) \cdot Z_j(t) + a(t) + s_i(t)\Big),
\end{split}
\end{equation}
where 
$$\sum_j \mathbbm{1}_i^{*,j}(t) = \sum_j \mathbbm{1}_i^{j}(t),$$ 
and 
$$\mathbbm{P}(\mathbbm{1}_i^{*,j_1}(t)=1)=\mathbbm{P}(\mathbbm{1}_i^{*,j_2}(t)=1) \quad\forall j_1,j_2 \in M.$$
Now, taking expectation of \eqref{eq:z_t_smart} and using its linearity, yields:
\begin{equation}\label{eq:z_t_exp_smart}
\begin{split}
    & \sum_j\E [Z_j(t+1)] \le \cr 
    & \sum_j\E[(1-\mathbbm{1}_i^{*,j}(t)\big) \cdot Z_j(t)] + m\lambda^{(1)} + m\mu_i^{(1)}.
\end{split}
\end{equation}
Next, using the law of total expectation, yields:
\begin{equation}\label{eq:z_t_exp_2_smart}
\begin{split}
  &\sum_j\E[(1-\mathbbm{1}_i^{*,j}(t)\big) \cdot Z_j(t)] = \cr 
  &\E\Big[\sum_j\E\big[(1-\mathbbm{1}_i^{*,j}(t)\big) \cdot Z_j(t) \,\big|\, Z_i(t)\big]\Big] =\cr
  & \E \Big[ \sum_j Z_j(t) \cdot \E\big[(1-\mathbbm{1}_i^{*,j}(t)\big) \,\big|\, Z_j(t)\big]\Big].
\end{split}
\end{equation}
Now, since $p > 0$, as in the proof of Proposition \ref{prop:2}, it immediately follows that there exists a strictly positive $\bar{\epsilon}$ such that $\E\big[\mathbbm{1}_i^{*,j}(t) \,\big|\, Z_j(t)\big] \ge p \cdot \bar{\epsilon}$. This means that, for all $j$
\begin{equation}\label{eq:z_t_exp_2_smart_2}
\begin{split}
  &\E[(1-\mathbbm{1}_i^{*,j}(t)\big) \cdot Z_j(t)] \le (1-p \cdot \bar{\epsilon})\E[Z_j(t)].
\end{split}
\end{equation}
Finally, using \eqref{eq:z_t_exp_2_smart_2} in \eqref{eq:z_t_exp_smart}, yields: 
\begin{equation}
\begin{split}
    & \sum_j\E [Z_j(t+1)] \le \cr & (1-p \cdot\bar{\epsilon})\sum_j\E[Z_j(t)] + m\lambda^{(1)} + m\mu_i^{(1)}.
\end{split}
\end{equation}
By denoting $Z(t) = \sum_j\E[Z_j(t)]$ we obtain 
\begin{equation}\label{eq:sum_z_subtit}
    Z(t+1) \le (1-p \cdot\bar{\epsilon})Z(t) + m\lambda^{(1)} + m\mu_i^{(1)}.
\end{equation}
Now using Lemma \ref{lem:gap_reccurence} in \eqref{eq:sum_z_subtit} yields
$$Z(t) \le \max \set{\frac{m\lambda^{(1)} + m\mu_i^{(1)}}{p \cdot\bar{\epsilon}},m \cdot C_0} \,\, \forall \, t.$$
Finally, since Assumption \ref{asm_2} holds for $Z(t)$, then it trivially holds for each of its positive components, \ie, $Z_j(t)$ for all $j \in M$. This concludes the proof.
\end{proof}


\section{Evaluation}

\T{Algorithms.} We proceed to present an evaluation study of three stable \name{} schemes, namely: \name{}-$Sample(2)$, \name{}-$Update(2m/n)$ and \name{}-$Smart(f^*(2m/n),\mbox{WAF})$. We compare them to the baseline full-information $JSQ$ and to the low-communication $JSQ(2)$ and $JIQ$. We note that all our three \name{} schemes are configured to have roughly the same expected communication overhead as the scalable $JSQ(2)$. Table \ref{tab:comp} summarizes the stability properties and the worst-case communication requirements of the evaluated load balancing techniques as established by our analysis and verified by our evaluations. 

\begin{table}[!t]
\centering  \scriptsize  
\begin{tabular}
{@{}lC{\mycc}C{\mycc}C{\mycc}C{\mycc}C{\mycc}C{\mycc}@{}} 
\toprule
                   & \multicolumn{2}{c}{Throughput opt.}  & \multicolumn{2}{c}{Comm. overhead} \\ \cmidrule{2-3} \cmidrule{4-5} 
                & \footnotesize{Homogeneous}       & \footnotesize{Heterogeneous}         & \footnotesize{Per time slot}  & \footnotesize{Per job arrival} \\  
\midrule
$JSQ$                                               & \checkmark  & \checkmark   &  $m \cdot n$     & $m$    \\ 
$JSQ(d)$                                            & \checkmark  & $\times$     &  $d \cdot m$     & $d$    \\ 
$JIQ$                                               & \checkmark  & $\times$     &  $n$             & $1$    \\ 
\midrule
\name{}-$Sample(d)$                                 & \checkmark  & \checkmark   &  $d \cdot m$     & $d$    \\ 
\name{}-$Update(p)$                      & \checkmark  & \checkmark   &  $n$             & $1$    \\ 
\name{}-$Update(f^*(p),\mbox{WAF})$      & \checkmark  & \checkmark   &  $n$             & $1$    \\ 
\bottomrule 
\end{tabular}
\caption{Comparing stability and \emph{worst-case} communication overhead of the evaluated load balancing techniques.} 
\label{tab:comp}
\vspace{-4mm}
\end{table}

\T{System.} In all our experiments, we consider a system of 100 servers and 10 dispatchers. \new{Recall that the system operates in time slots with the following order of events within each time slot: (1)~jobs arrive at each dispatcher; (2)~a routing decision is taken by each dispatcher and it immediately forwards its jobs to one of the servers; (3)~each server performs its service for this time-slot.}

\T{Arrivals.} \new{The number of jobs that arrive at each dispatcher at each time slot is sampled from a Poisson distribution with parameter $0.1 \cdot \lambda$, hence the total number of arrivals to the system at each time slot is a sample from a Poisson distribution with parameter $\lambda$. Each of the $10$ dispatchers does not store incoming jobs, and instead immediately forwards them to one of the $100$ servers.}

\T{Departures.} \new{Each server has an unbounded FIFO queue for storing incoming jobs.} We divide the servers into the following two groups: (1) weak and (2) strong. We then consider different mixes of their numbers: (1) 10\% strong - 90\% weak; (2) 50\% strong - 50\% weak; (3) 90\% strong - 10\% weak. We also consider two different cases of heterogeneity (that is, service rate ratio between a strong and a weak server): (1) moderate heterogeneity with a 1:2 service rate ratio; (2) high heterogeneity with a 1:10 service rate ratio. 


\subsection{Moderate heterogeneity}

\begin{figure*}[t!]
\centering
\begin{subfigure}{0.30\linewidth}
\includegraphics[width=\textwidth]{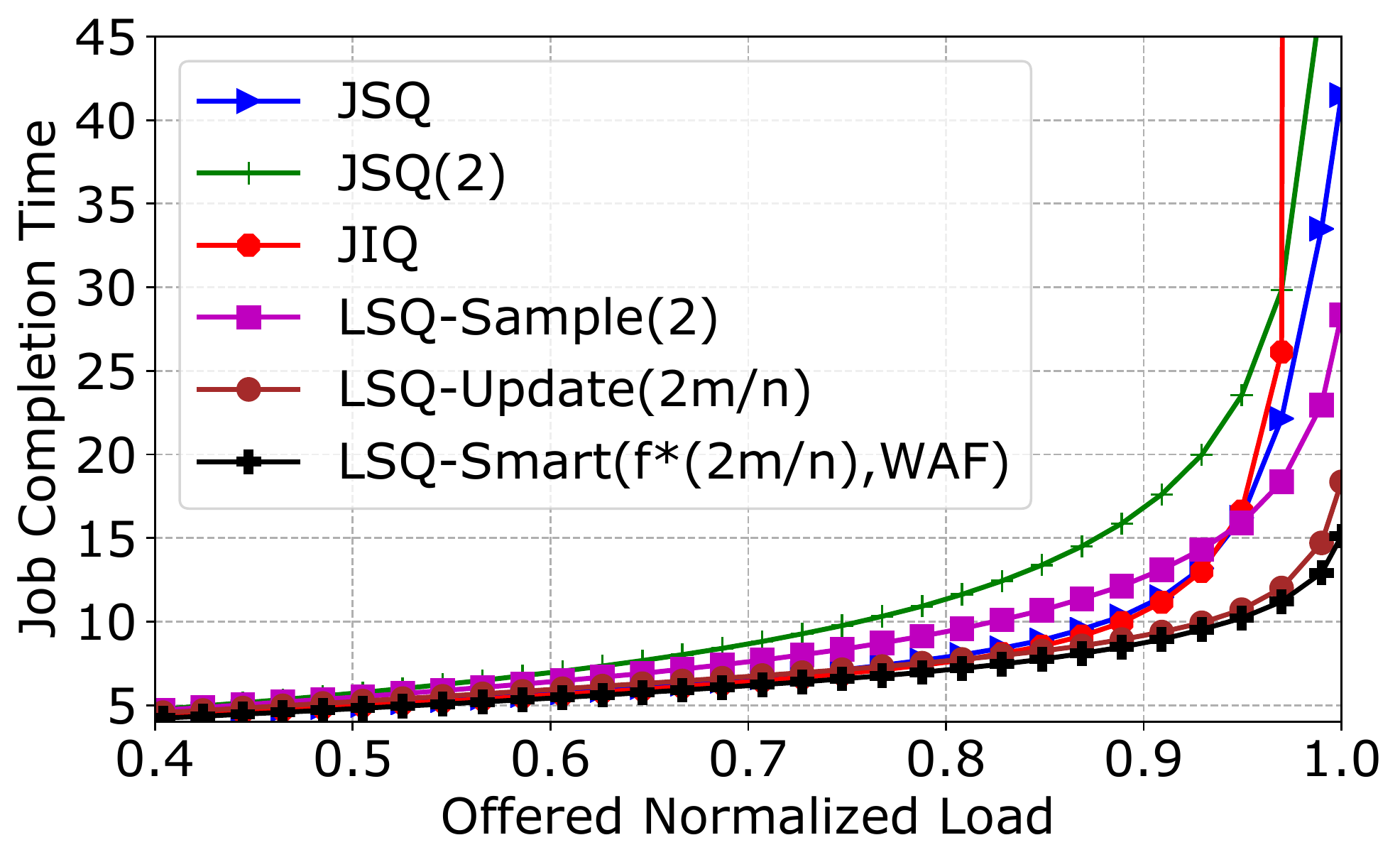}
  \caption{10 strong servers, 90 weak servers.}
  \label{fig:mod_het:jct:01}
\end{subfigure}
\begin{subfigure}{0.30\linewidth}
  \includegraphics[width=\textwidth]{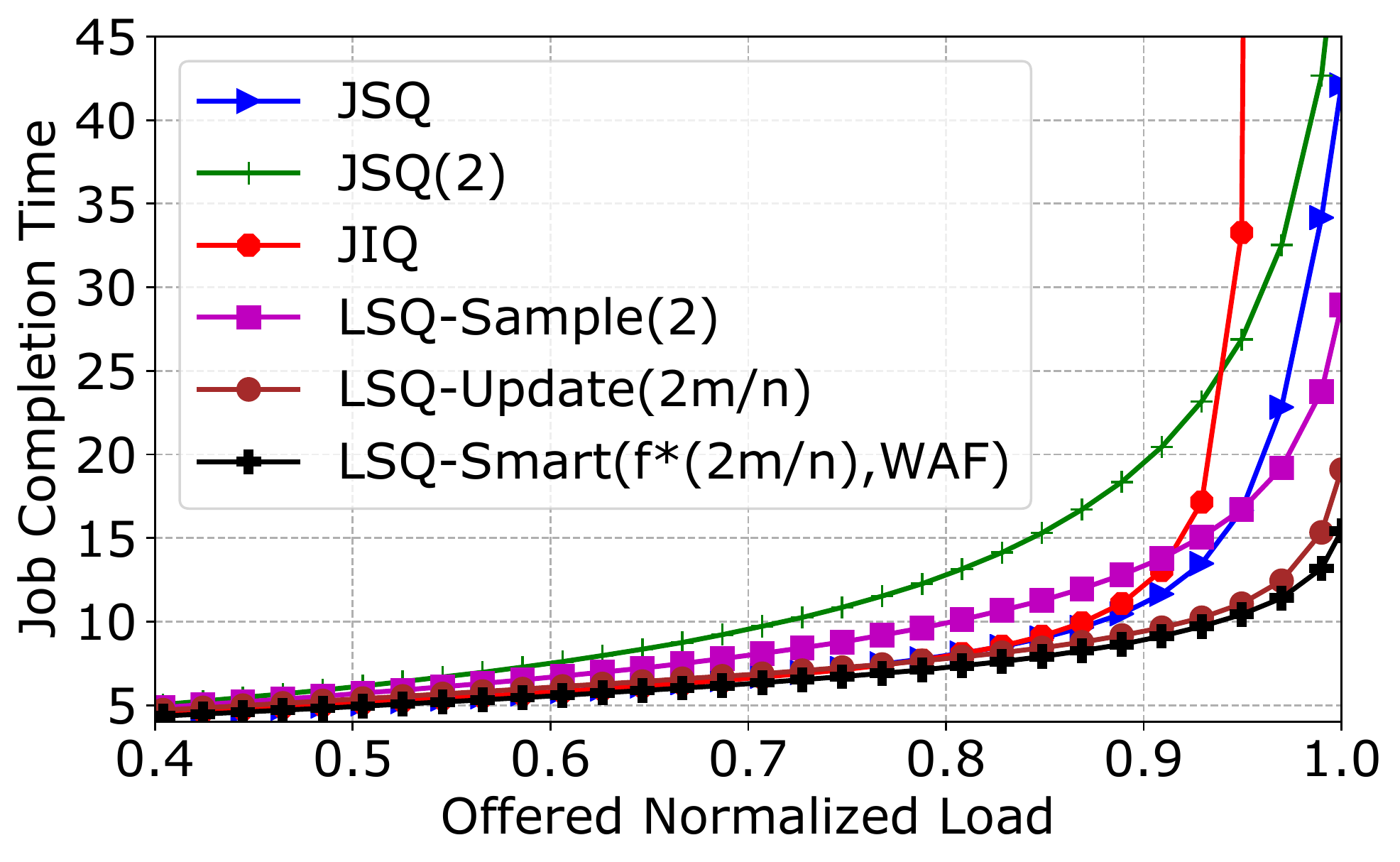}
  \caption{50 strong servers, 50 weak servers.}
  \label{fig:mod_het:jct:05}
\end{subfigure}
\begin{subfigure}{0.30\linewidth}
  \includegraphics[width=\textwidth]{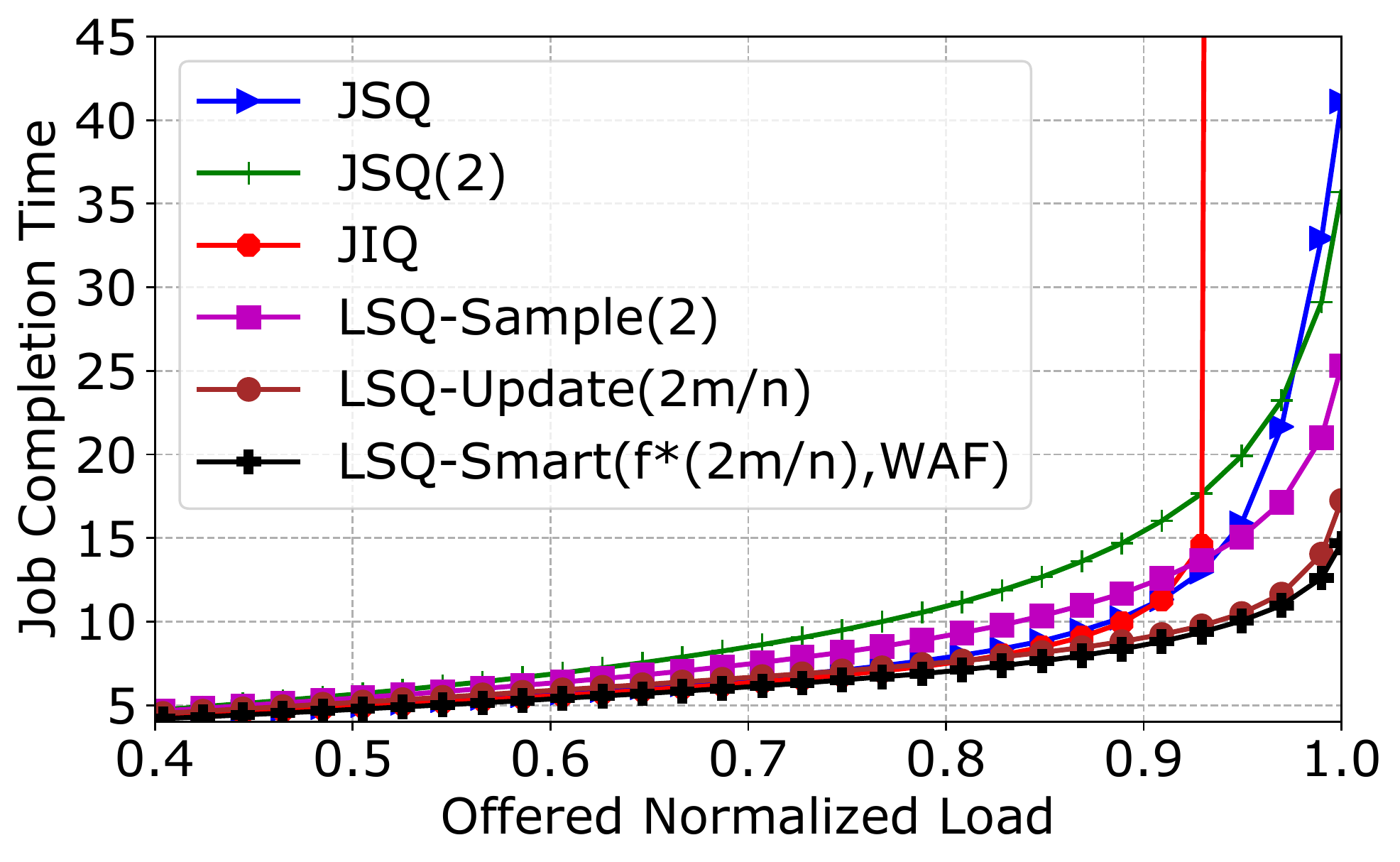}
  \caption{90 strong servers, 10 weak servers.}
  \label{fig:mod_het:jct:09}
\end{subfigure}
\begin{subfigure}{0.30\linewidth}
\includegraphics[width=\textwidth]{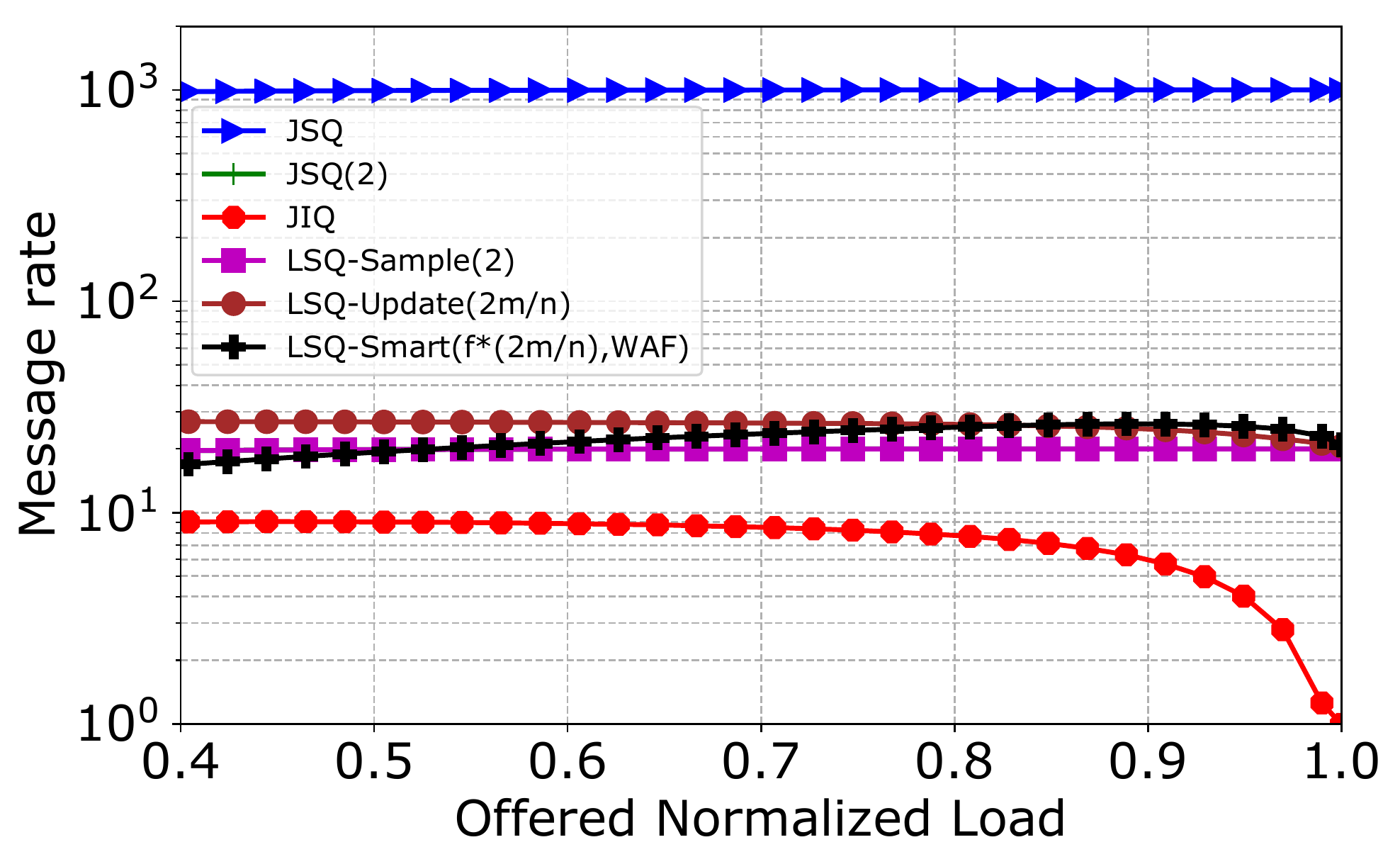}
  \caption{10 strong servers, 90 weak servers.}
  \label{fig:mod_het:mess:01}
\end{subfigure}
\begin{subfigure}{0.30\linewidth}
  \includegraphics[width=\textwidth]{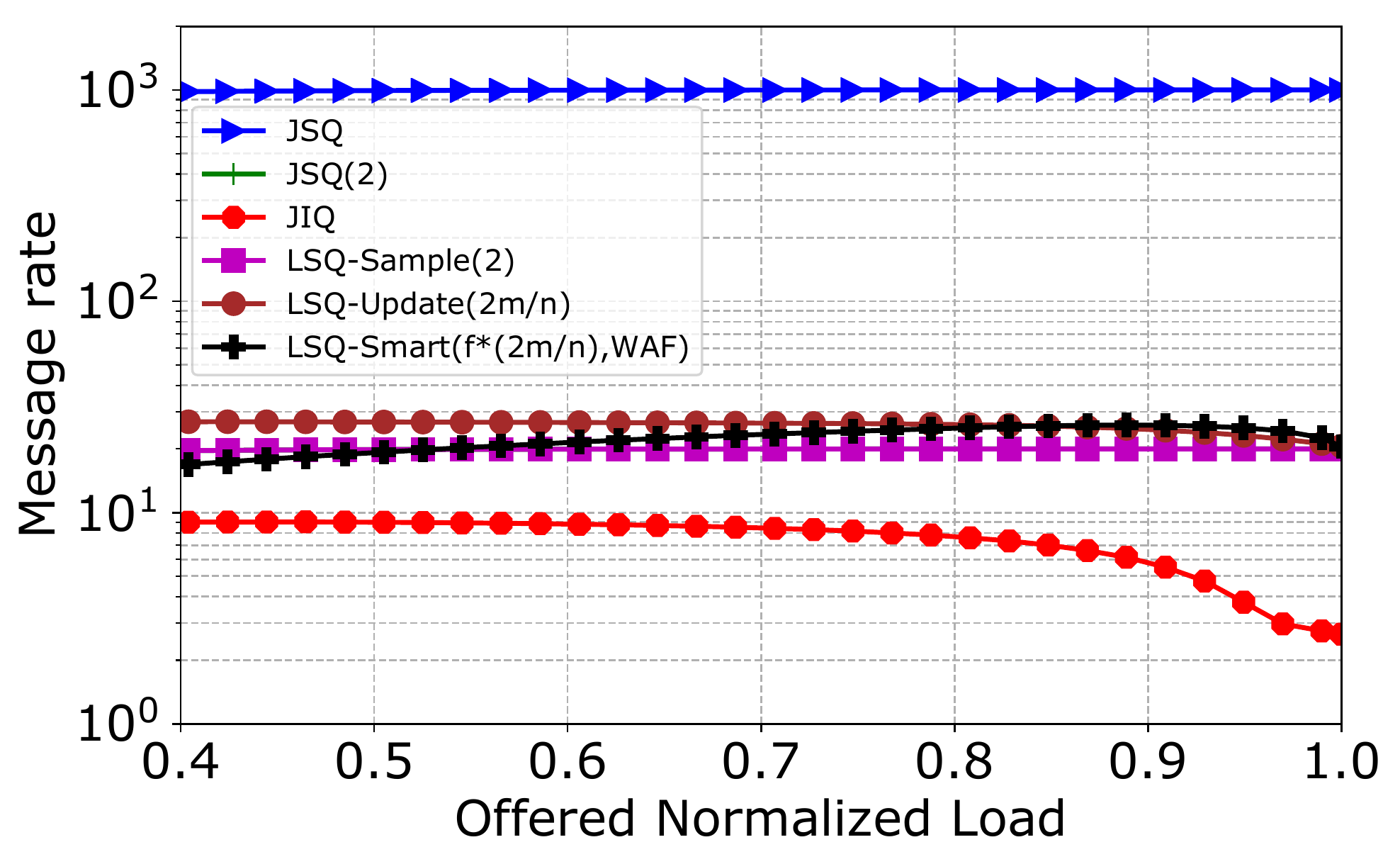}
  \caption{50 strong servers, 50 weak servers.}
  \label{fig:mod_het:mess:05}
\end{subfigure}
\begin{subfigure}{0.30\linewidth}
  \includegraphics[width=\textwidth]{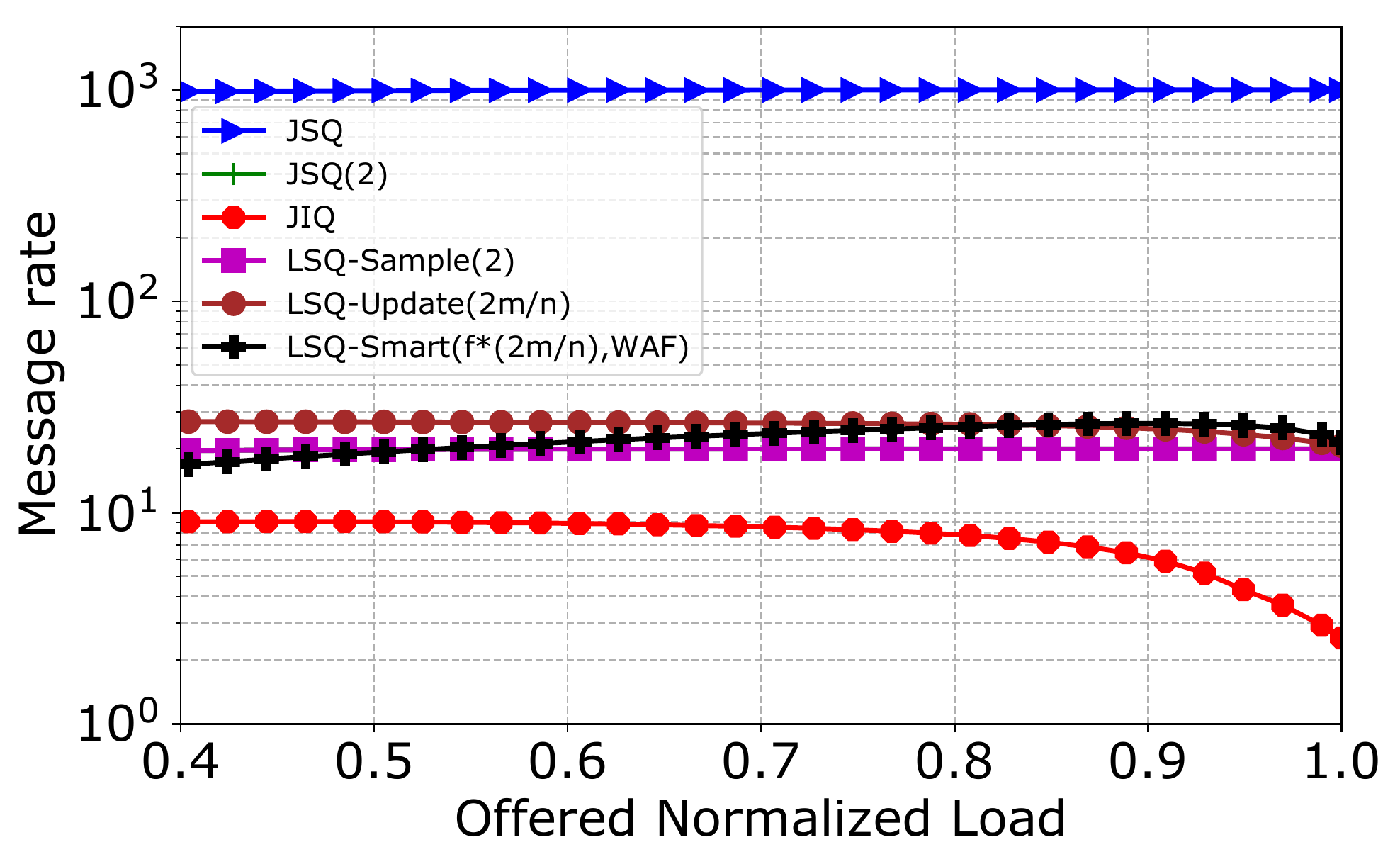}
  \caption{90 strong servers, 10 weak server.}
  \label{fig:mod_het:mess:09}
\end{subfigure}
\caption{Moderate heterogeneity scenario with 10 dispatchers and 100 heterogeneous servers.}
\label{fig:mod_het}
\end{figure*}

We start with a moderate degree of heterogeneity. Specifically, the service processes are geometrically distributed with a parameter $2p$ for a weak server and a parameter $p$ for a strong server. \new{That is, at each time slot, the number of jobs that may be completed by each server is sampled from a geometric distribution with its respective parameter}. In a simulation with $n_s$ strong servers and $n_w$ weak servers, we set $p = \frac{n_s + 0.5n_w}{100}$ and sweep $0 \le \lambda < 100$. The results are presented in Figure \ref{fig:mod_het}.

\T{Stability.} It should be noted that, in all three scenarios, $JIQ$ is not stable whereas all other algorithms are. Also, $JIQ$'s stability region decreases as the number of strong servers in the mix increases. Intuitively, $JSQ(2)$ is stable in these scenarios because a strong server has to receive 2$\times$ more jobs than a weak server at high loads in order to obtain stability. Since in $JSQ(2)$ a dispatcher samples 2 servers at each arrival, the probability in which a strong server is sampled is sufficiently high (\eg 0.75 for the 50\%-50\% case in Figure \ref{fig:mod_het:jct:05}). This, in turn, is sufficient to obtain stability. As implied by mathematical analysis, $JSQ$ and all three $LSQ$ variant are stable.

\T{Performance.} In all three scenarios and over all loads, our pull-based schemes exhibit the best performance, whereas at high loads our push-based scheme outperforms $JSQ$ as well. This is because $JSQ$ suffers from an increasing incast effect as the load increases. That is, at high loads, there is usually only a single shortest queue that receives all incoming work for that time slot. At low loads, on the other hand, there are usually several idle servers that are randomly picked by the dispatchers, hence the incast is less significant. The performance of $JSQ(2)$ is significantly worse than our policies at low and moderate loads, whereas at high loads it is inconsistent and degrades as the number of strong servers in the mix decreases. 

\T{Communication overhead.} As expected, all our $LSQ$ schemes incur roughly the same communication overhead as the scalable $JSQ(2)$ policy. Remarkably, this is two orders of magnitude less than $JSQ$. Even better communication overhead is achieved by $JIQ$ but, as mentioned, it is not a stable policy.  


\subsection{High heterogeneity}

\begin{figure*}[t!]
\centering
\begin{subfigure}{0.30\linewidth}
\includegraphics[width=\textwidth]{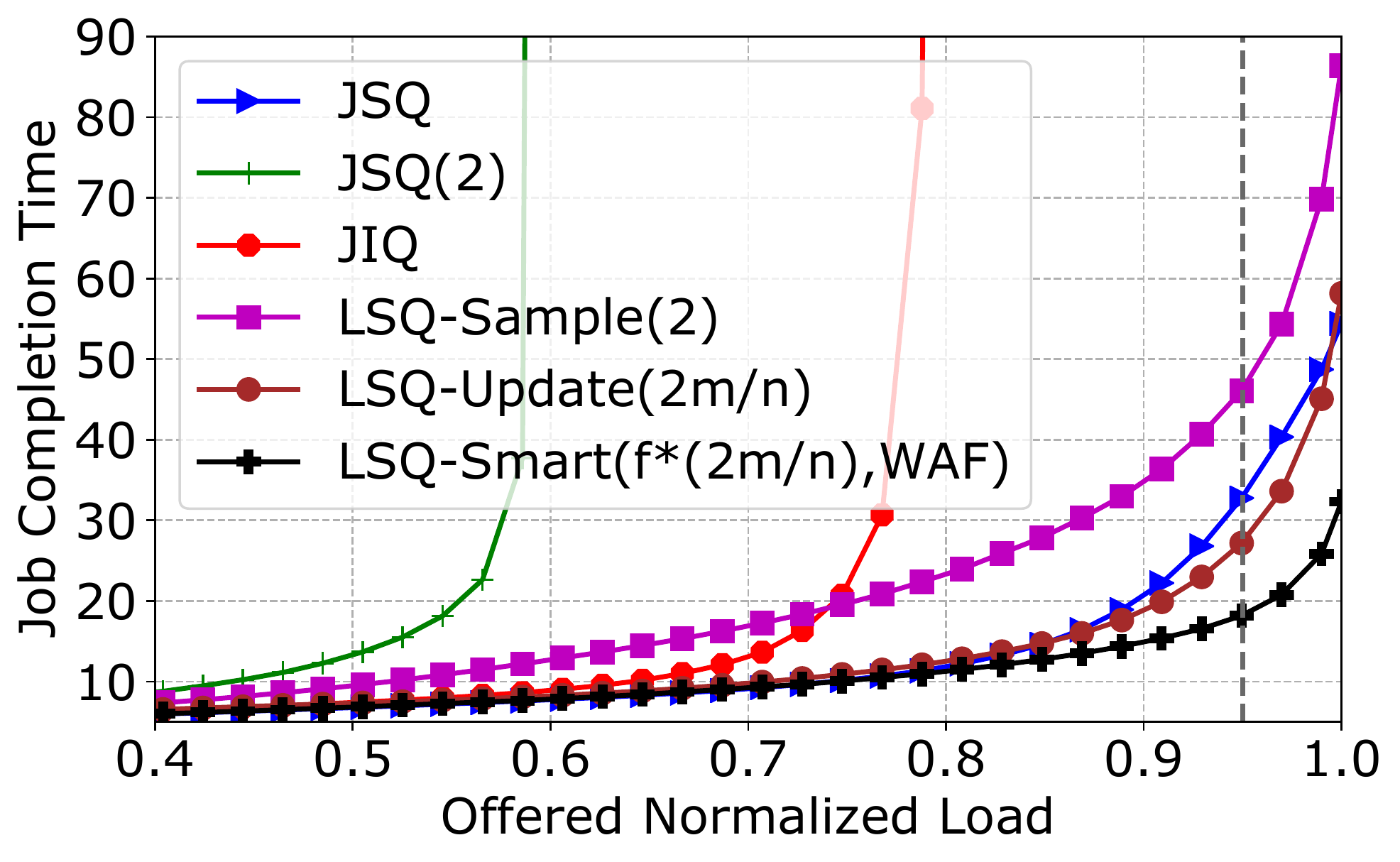}
  \caption{10 strong servers, 90 weak servers.}
  \label{fig:high_het:jct:01}
\end{subfigure}
\begin{subfigure}{0.30\linewidth}
  \includegraphics[width=\textwidth]{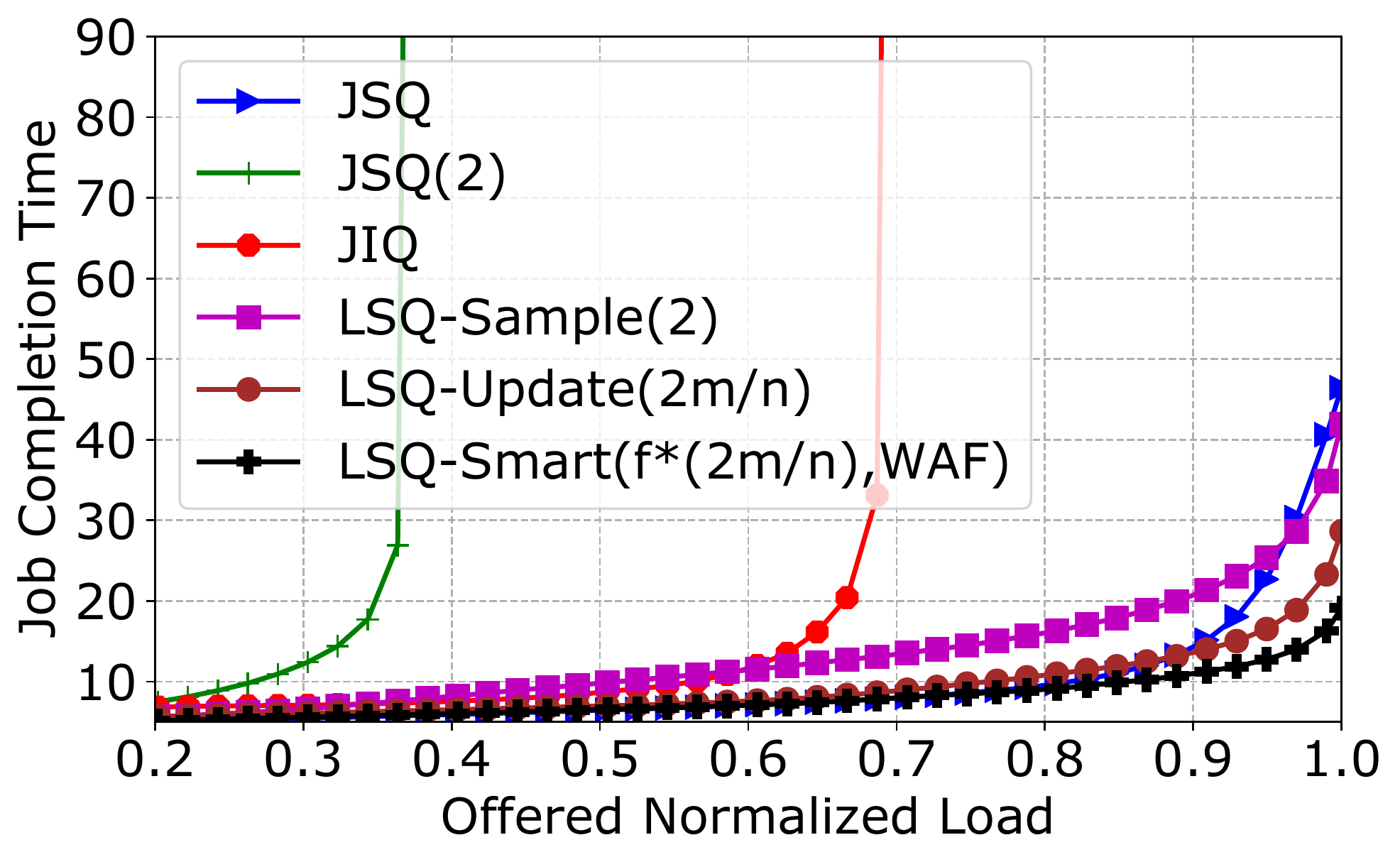}
  \caption{50 strong servers, 50 weak servers.}
  \label{fig:high_het:jct:05}
\end{subfigure}
\begin{subfigure}{0.30\linewidth}
  \includegraphics[width=\textwidth]{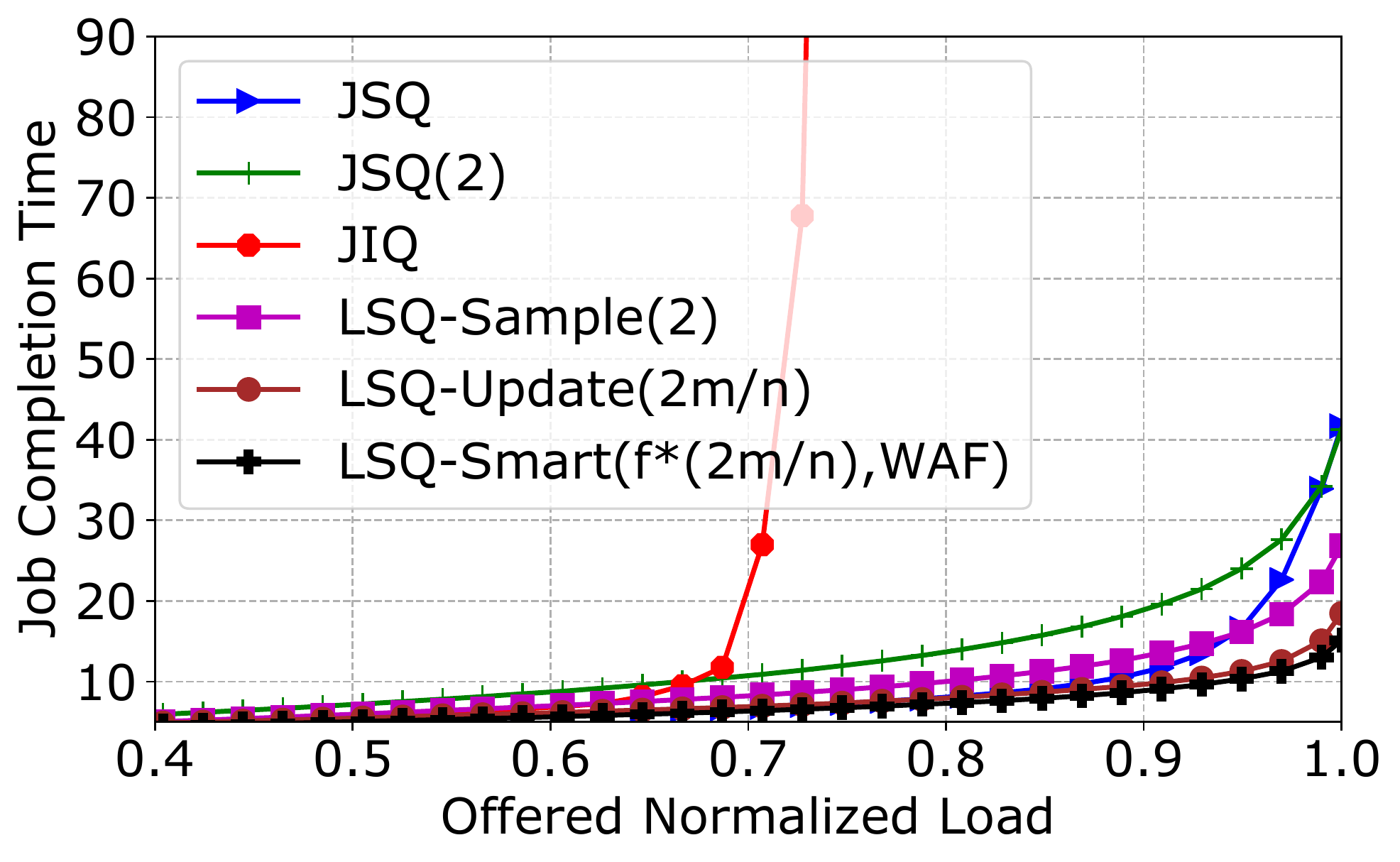}
  \caption{90 strong servers, 10 weak servers.}
  \label{fig:high_het:jct:09}
\end{subfigure}
\begin{subfigure}{0.30\linewidth}
\includegraphics[width=\textwidth]{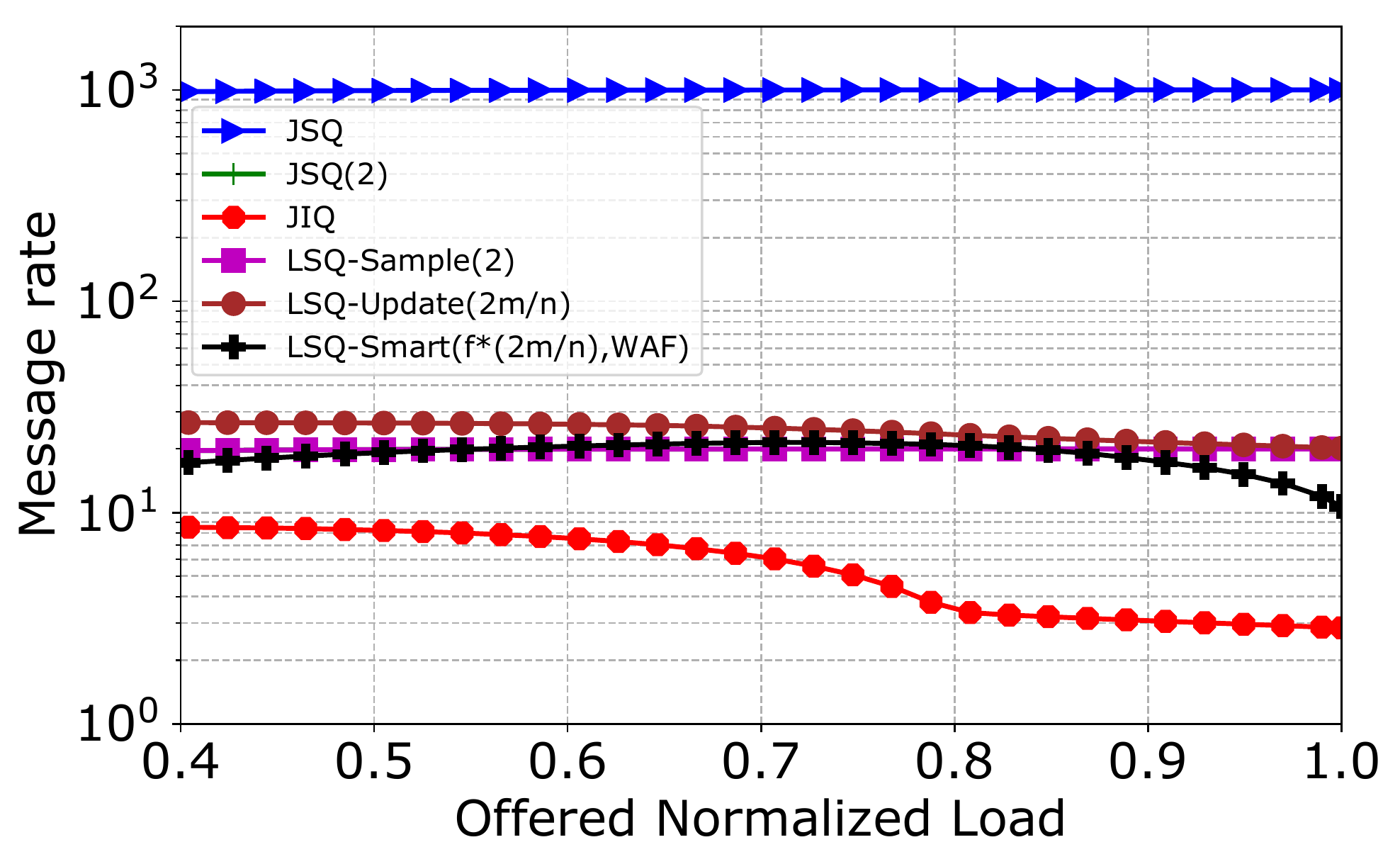}
  \caption{10 strong servers, 90 weak servers.}
  \label{fig:high_het:mess:01}
\end{subfigure}
\begin{subfigure}{0.30\linewidth}
  \includegraphics[width=\textwidth]{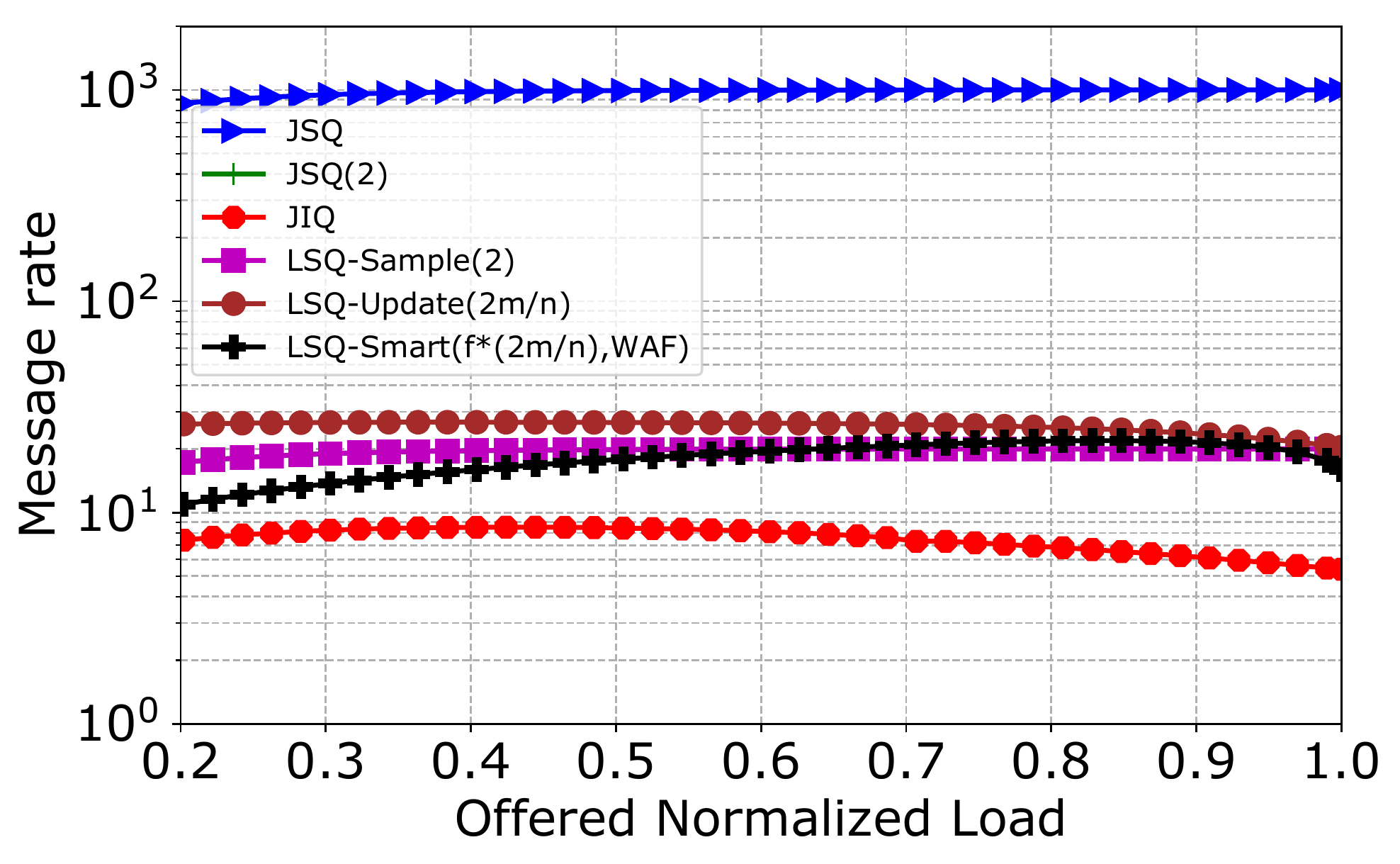}
  \caption{50 strong servers, 50 weak servers.}
  \label{fig:high_het:mess:05}
\end{subfigure}
\begin{subfigure}{0.30\linewidth}
  \includegraphics[width=\textwidth]{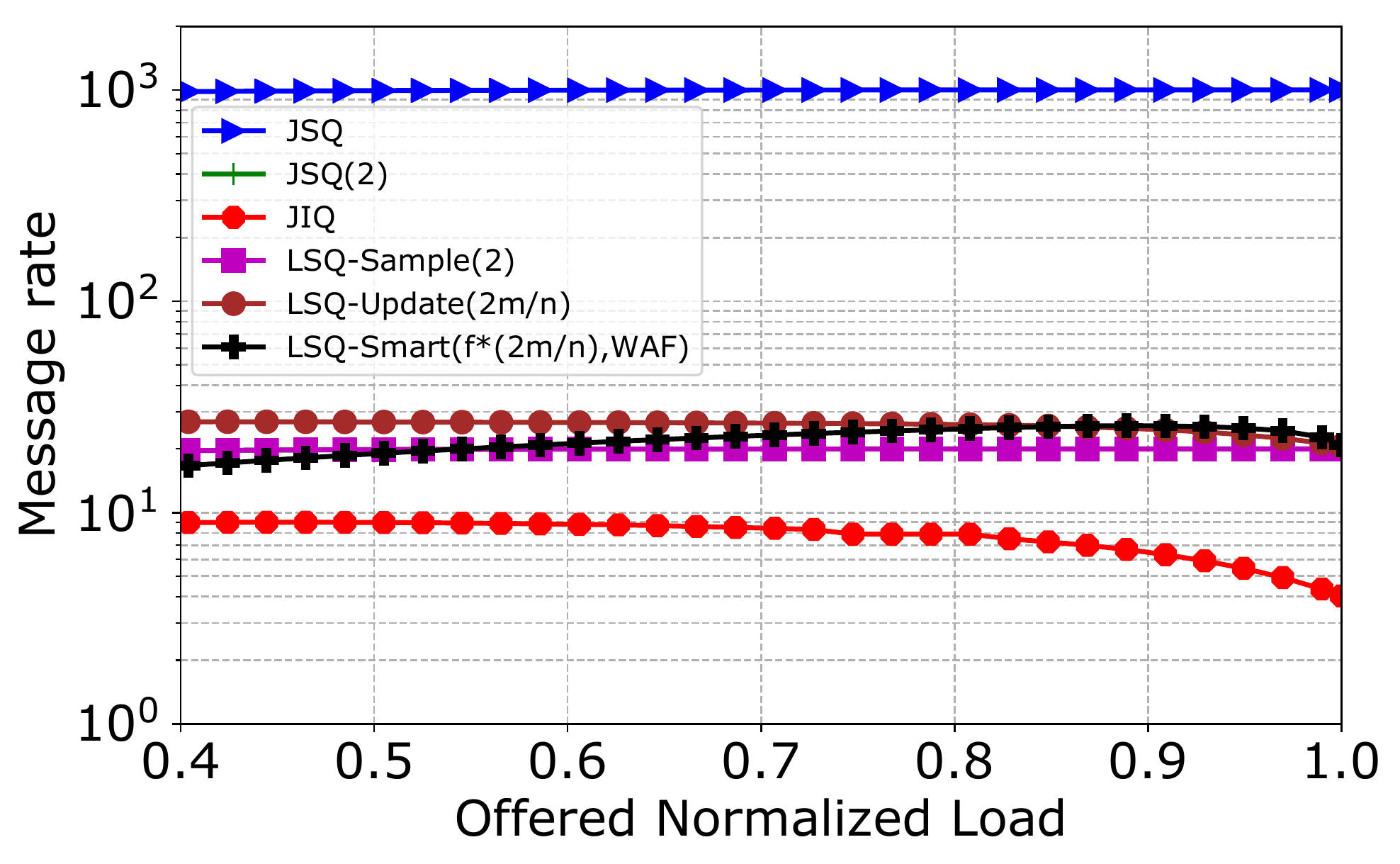}
  \caption{90 strong servers, 10 weak server.}
  \label{fig:high_het:mess:09}
\end{subfigure}
\caption{High heterogeneity scenario with 10 dispatchers and 100 heterogeneous servers.}
\label{fig:high_het}
\end{figure*}

We proceed to examine scenarios with a high degree of heterogeneity. Specifically, in these scenarios, the server service processes are geometrically distributed with a parameter $10p$ for a weak server and a parameter $p$ for a strong server. In a simulation with $n_s$ strong servers and $n_w$ weak servers, we set $p = \frac{n_s + 0.1n_w}{100}$ and sweep $0 \le \lambda < 100$. The results are presented in Figure \ref{fig:high_het}.

\T{Stability.} Again, in all three scenarios, $JIQ$ is not stable. Also, $JIQ$'s stability region is significantly decreased due to the higher levels of heterogeneity. $JSQ(2)$ is unstable as well, with even worse degradation in the stability region. Specifically, it is stable only when there are only 10\% weak servers in the mix such that the probability of not sampling a strong server upon arrival is sufficiently low. Again, as implied by mathematical analysis, $JSQ$ and all our three $LSQ$ schemes are stable.

\T{Performance.} Again, in all three scenarios and over all loads, our pull-based schemes exhibit the best performance. At high loads, our push-based scheme outperforms $JSQ$ in two out of the three scenarios, whereas it under-performs when there is a low number of strong servers in the mix. 

\T{Communication overhead.} Again, as expected, all our $LSQ$ schemes incur roughly the same communication overhead as the unstable $JSQ(2)$ policy. Recall that this is two orders of magnitude less than $JSQ$. 

\T{Delay tail distribution.} Another finding of our simulation results is that \emph{all} of the three $LSQ$ policies consistently provide a better delay tail distribution than $JSQ$. For example, in Figure~\ref{fig:high_het:jct:01:CCDF}, we present the CCDF of all stable policies at a normalized load of 0.95 in the scenario of Figure~\ref{fig:high_het:jct:01} (marked by a dashed grey line). It should be noted that $JSQ$ has a lower average job completion time than our $LSQ$-$Sample(2)$ policy. Nevertheless, $JSQ$ has a worse delay tail distribution. As illustrated in Figure~\ref{fig:high_het:jct:01:incast}, this is due to the incast effect, where the majority (or even all) of the dispatchers forward their incoming jobs to a single (least loaded) server. It is notable how in all three \name{} policies the incast is nearly eliminated by allowing different dispatchers to have a different view of the system.

\newer{Interestingly, Figure~\ref{fig:high_het:jct:01:CCDF} also shows that, for a small fraction that accounts for less than 0.00001 of the jobs (i.e., a rare event), the completion time under $LSQ{-}Smart(f^*(2m/n,WAF))$ is slightly larger than under $LSQ{-}Update(2m/n)$. Namely, even though $LSQ{-}Smart(f^*(2m/n,WAF))$ has a lower mean job completion time, it has a slightly slower decreasing delay tail. Essentially, the reason for this phenomenon is that even though for both policies incast is rare, in this specific scenario for $LSQ{-}Smart(f^*(2m/n,WAF))$, it is less rare than for $LSQ{-}Update(2m/n)$. For example, when using $LSQ{-}Smart(f^*(2m/n,WAF))$, for a small time fraction of $\approx$0.0001, five dispatchers send their jobs to the same server. But for $LSQ{-}Update(2m/n)$, it decreases by order of magnitude to $\approx$0.00001.}

\begin{figure}[t!]
\centering
\begin{subfigure}{0.85\linewidth}
\includegraphics[width=\textwidth]{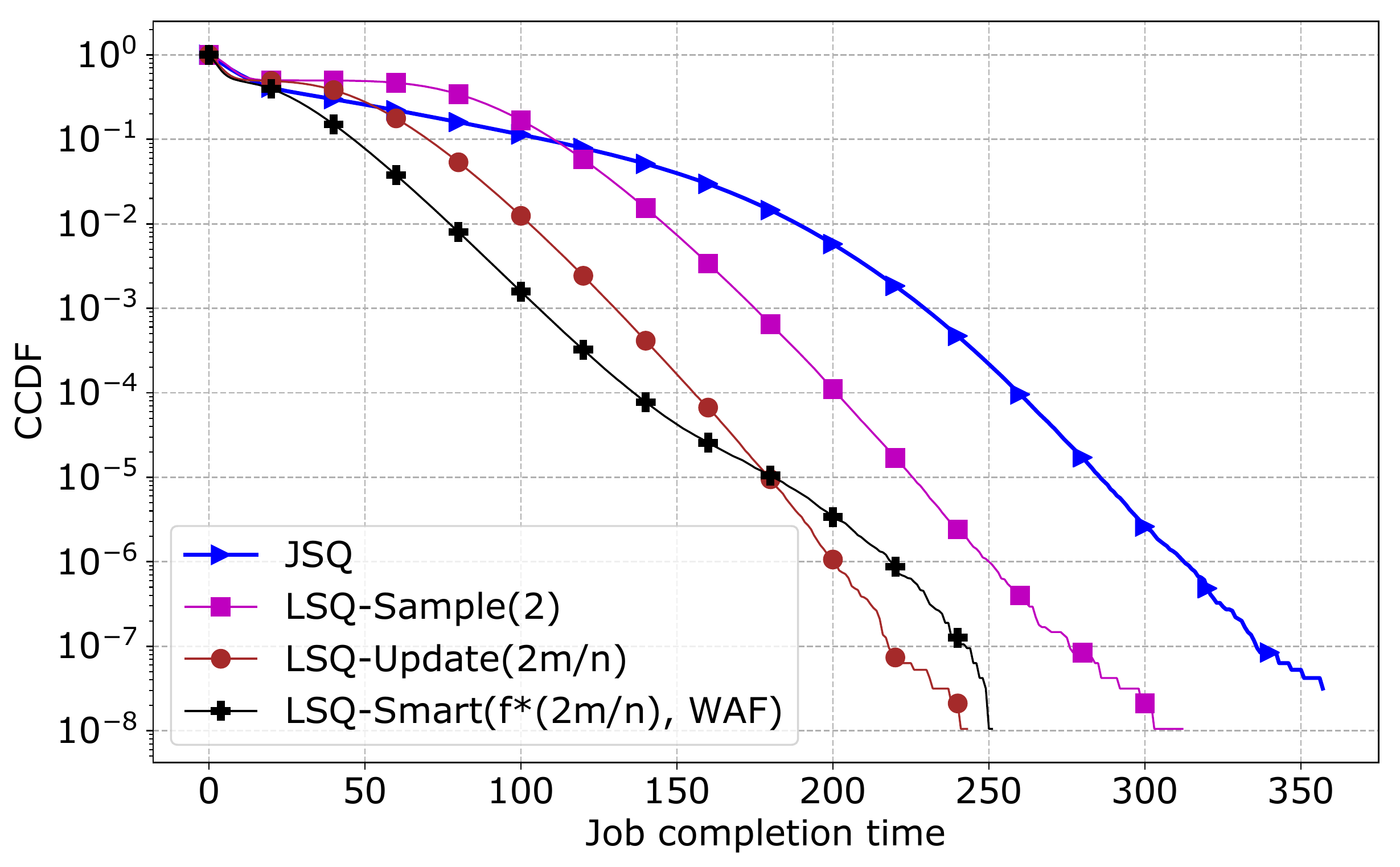}
  \caption{Job completion time CCDF.}
  \label{fig:high_het:jct:01:CCDF}
\end{subfigure}
\begin{subfigure}{0.85\linewidth}
  \includegraphics[width=\textwidth]{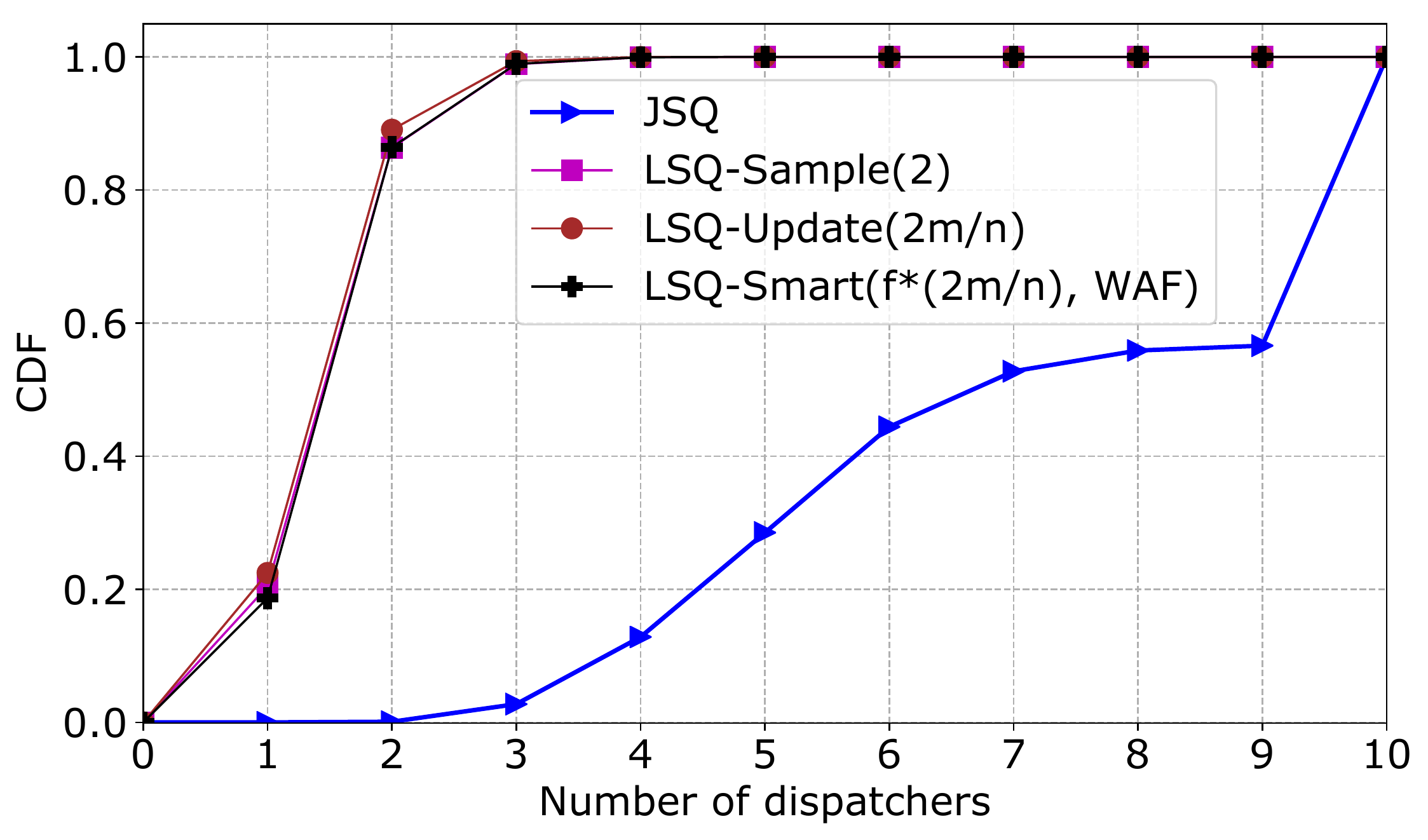}
  \caption{Maximum number of dispatchers that forward their arriving jobs to the same server at a time slot.}
  \label{fig:high_het:jct:01:incast}
\end{subfigure}

\caption{Illustration of the job completion time CCDF and the incast effect at a normalized load of 0.95 for the scenario in Figure~\ref{fig:high_het:jct:01} (marked by a dashed grey line).}
\vspace{-5mm}
\label{fig:high_het:jct:01:load_095}
\end{figure}


\subsection{Evaluation takeaways}

The three tested \name{} approaches always guarantee stability and do so using roughly the same communication budget as the non-throughput-optimal $JSQ(2)$. \new{Moreover}, the simulations indicate that, under these low-communication requirements, the tested \name{} policies consistently exhibit good performance in different scenarios and even admit better performance and delay tail distribution than the full-state information $JSQ$ that uses more communication overhead by orders of magnitude. 

Additionally, the evaluation results indicate how having pull-based communication and smart servers can further improve performance while using a similar communication budget. This is because pull-based communication and smart servers allow us to tune the system towards sending more messages from less loaded servers and directing them to less updated dispatchers, hence making better use of the communication budget. 

\T{Remark.} \newest{In this paper, we target heterogeneous servers. Nevertheless, we have simulation results that indicate that LSQ techniques provide better performance in terms of job completion times for the homogeneous case (where all the aforementioned techniques are stable) as well. Furthermore, they indicate that, in the homogeneous case, JSQ(2) offers better performance than JSQ at high loads and even delay tails that are competitive with LSQ due to the reduced incast. Due to space limits, these results have not been included in this paper. However, our evaluation code, which we intend to release upon the publication of the paper, allows to fully and easily recreate these results as well.}


\section{Arbitrarily low communication}

We have shown, both formally and by way of simulations, how different \name{} schemes offer strong theoretical guarantees and appealing performance with low communication overhead. In particular, by virtue of Theorems \ref{thm:simple_cond} and \ref{thm:simple_cond_2}, we can construct various strongly stable \name{} policies with any arbitrarily low communication budget, disregarding whether the system uses pull or push messages (or both). 

Achieving strong stability with arbitrary low communication is known to be possible with homogeneous servers, since even a uniform random load-balancing policy is stable in that case, and this was indeed strengthened in \cite{van2019hyper} for a scheme similar to \name{} and a different stability criterion \emph{in continuous time}. However, establishing this for heterogeneous servers and multiple dispatchers is far from straightforward, and constitutes one of the main contributions of this paper.

Let $M(t)$ be the number of queue length updates performed by all dispatchers up to time $t$. Fix any arbitrary small $r>0$. Suppose that we want to achieve strong stability, such that the average message rate is at most $r$, \ie for all $t$ we have that $\E [M(t)] \le rt$. Then, the two following per-time-slot dispatcher sampling rules trivially achieve strong stability (by Theorem \ref{thm:simple_cond}) and respect the desired bound, \ie $\E [M(t)] \le rt$. 

\begin{example}[push-based communication example] Dispatcher sampling rule upon job(s) arrival:
\begin{enumerate}
    {\setlength\itemindent{15pt} \item[(1)] pick a server $i \in N$ uniformly at random.}
    {\setlength\itemindent{15pt} \item[(2)] sample server $i$ with probability $\frac{r}{m}$.}
\end{enumerate}
\end{example}

\begin{example}[pull-based communication example]
Server messaging rule upon job(s) completion:
\begin{enumerate}
{\setlength\itemindent{15pt} \item[(1)] pick a dispatcher $j \in M$ uniformly at random.}
{\setlength\itemindent{15pt} \item[(2)] update dispatcher $j$ with probability $\frac{r}{n}$.}
\end{enumerate}

\end{example}
These theorems also enable us to design stable \name{} policies with hybrid communication (\eg push and pull) that attempt to maximize the benefits of both approaches. For example, the following policy leverages both the advantages of pull-based communication (\ie being immediately notified that a server becomes idle) and push-based communication (\ie random exploration of shallow queues when no servers are idle).

\begin{example}[Hybrid communication example]
Dispatcher sampling rule: 
\begin{enumerate}
{\setlength\itemindent{15pt} \item[(1)] pick a server $i \in N$ uniformly at random.} 
{\setlength\itemindent{15pt} \item[(2)] sample server $i$ with probability $\frac{r}{m}$.} 
\end{enumerate}
Server messaging rule: 
\begin{enumerate}
{\setlength\itemindent{15pt} \item[(1)] if got idle, pick a dispatcher $j \in M$ uniformly at random and send it an update message.} 
\end{enumerate}

\end{example}

The above examples demonstrate the wide range of possibilities that the \name{} approach offers to the design of stable, scalable policies with arbitrarily low communication overhead. 


\section{Discussion}

Before concluding the study, we discuss several additional properties of the different load balancing approaches.

\T{Instantaneous routing.} An appealing property of any \name{} policy, similarly to $JIQ$, is that a dispatcher can immediately take routing decisions upon a job arrival. This is in contrast to common push-based policies that have to wait for a response from the sampled servers to be able to make a decision. For example, when using the $JSQ(2)$ policy, when a job arrives the dispatcher cannot immediately send the job to a server but must pay the additional delay of sampling two servers.

\T{Space requirements.} To implement an \name{} policy, similarly to $JSQ$, each dispatcher has to hold an array of size $n$ (local views). When smart servers are used, each server also has to hold an array of size $m$ (dispatcher states). Such a space requirement incurs negligible overhead on a modern server. For example, nowadays, any commodity server has tens to hundreds of GBs of DRAM. But even a hypothetical cluster with $10^6$ servers requires only a few MB of the dispatcher's memory and much less (at least by 1-2 orders of magnitude) of server's memory, which is negligible in comparison to the DRAM size.

\T{Computational complexity.} To implement an \name{} policy, similarly to $JSQ$, each dispatcher has to repeatedly find the minimum (local) queue length. By using a priority queue (\eg min-heap), finding the minimum results in only a single operation (i.e., simply looking at the head of the priority queue). For a queue length update operation, $O(\log n)$ operations are required in the worst case (\eg decrease-key operation in a min-heap).\footnote{A more sophisticated data structure, such as the Fibonacci heap, may offer even O(1) operations per update yet incur a higher constant.} Even with $n=10^6$, just a few operations are required in the worst case per queue length update. This results in a single commodity core being able to perform tens to hundreds of millions of such updates per second, hence resulting in negligible overhead, especially for a low-communication policy in which queue length updates are not too frequent.

\T{\new{Inaccurate information can lead to better performance.}} \new{Although all tested \name{} variants in this paper are proved to be throughput-optimal, it is still surprising that using inaccurate information can lead to better performance than $JSQ$. In fact, we have found that allowing outdated information in the multi-dispatcher scenario not only reduces communication overhead significantly but also often results in better performance when compared to the full-state $JSQ$. This is because the incast effect can significantly degrade the performance of $JSQ$  when many dispatchers forward their jobs to the shortest queue(s) simultaneously. On the other hand, as the simulation results indicate, having inaccurate information reduces the incast effect, since each dispatcher may believe that a different queue is the shortest (\eg Fig. \ref{fig:high_het:jct:01:CCDF} and \ref{fig:high_het:jct:01:incast}). Intuitively, in \name{}, by allowing inaccurate local states of the queue lengths, jobs are being forwarded to queues that may not be the least loaded ones but still have low load, which therefore reduces incast.}

\T{Alleviating incast with noise.} Another natural way to reduce incast may be to use the $JSQ$ policy with some sophisticated i.i.d. noise addition scheme that locally modifies the approximated queue lengths at each dispatcher in order to break synchronization while preserving performance. In fact, we can reuse our proof of Theorem \ref{thm:1} to show that such a scheme is ensured to be strongly stable if the added noise is bounded in expectation. \new{In a sense, this is related to the issue discussed above, namely, that inaccurate information can lead to better performance.}
Nonetheless, such a \new{$JSQ$-based} solution is still not scalable in terms of communication overhead.


\section{Conclusion}

In this paper, we introduced the \name{} family of load balancing algorithms. We formally established an easy-to-satisfy sufficient condition for an \name{} policy to be strongly stable. We further developed easy-to-verify sufficient stability conditions and exemplified their use. Then, using simulations, we showed how different \name{} schemes significantly outperform well-known low-communication policies, such namely $JSQ(d)$ and $JIQ$, while consuming a similar communication budget. We further demonstrated how relying on pull-based communication and, even further, on smart servers, allows \name{} to outperform even $JSQ$ in terms of both the means and tail distributions of the job completion times, while using orders of magnitude less communication. 

\section{\new{Future work}}

\T{\new{Performance guarantees.}} \new{In this work, we obtained the throughput optimality of different scalable \name{} policies for heterogeneous systems with multiple dispatchers. We believe that the theoretical investigation of specific performance guarantees (\eg delay bounds) may lead to even better \name{} techniques.
}  

\T{\new{Adversarial evaluations.}} \new{In this work, we employed simulations in order to test different \name{} policies by way of simulations with respect to our considered system model, and in comparison to other proposed policies. It will be of interest to test the \name{} concept on real systems as well as to explore how \name{} performs under different adversarial settings in comparison to current practice.} 

\T{\newest{Known service rates.}} \newest{In this work, we assume that the servers service rates are unknown to the dispatchers. However, in systems where full or partial knowledge regarding these rates can be obtained, it is of interest to investigate how such knowledge can be used in order to provide improved load balancing solutions to the many-dispatcher case.}

\section*{Acknowledgments}

The authors would like to thank associate editor Paolo Giaccone and the anonymous reviewers for their helpful comments. This work was partly supported by the Hasso Plattner Institute Research School, the Israel Ministry of Science and Technology, the Technion Hiroshi Fujiwara Cyber Security Research Center, the Israel Cyber Bureau, and the Israel Science Foundation (grant No. 1119/19).


\bibliographystyle{abbrv}
\bibliography{refs}

\begin{thebibliography}{10}

\bibitem{lsqcode}
{LSQ} code.
\newblock \url{https://github.com/LocalShortestQueue/LocalShortestQueue}.

\bibitem{Openreviewlsq}
Openreview preprint server. {Load Balancing in Large-Scale Heterogeneous
  Systems with Multiple Dispatchers}. anonimous preprint, 2018.
\newblock \url{https://openreview.net/pdf?id=BJxN9Hvf-V}.

\bibitem{shayvthesys}
Vargaftik {S}hay. {S}cheduling and {L}oad {B}alancing in {E}merging {N}etworked
  {S}ystems. {P}h.{D}. thesis, 2019.
\newblock
  \url{https://www.graduate.technion.ac.il/Theses/Abstracts.asp?Id=29940}.

\bibitem{ptp}
{IEEE Standard for a Precision Clock Synchronization Protocol for Networked
  Measurement and Control Systems}.
\newblock {\em IEEE Std 1588-2008 (Revision of IEEE Std 1588-2002)}, pages
  1--269, July 2008.

\bibitem{jonatha2018power}
J.~Anselmi and F.~Dufour.
\newblock Power-of-$ d $-choices with memory: Fluid limit and optimality.
\newblock {\em arXiv preprint arXiv:1802.06566}, 2018.

\bibitem{bramson2010randomized}
M.~Bramson, Y.~Lu, and B.~Prabhakar.
\newblock Randomized load balancing with general service time distributions.
\newblock In {\em ACM SIGMETRICS}, 2010.

\bibitem{bramson2012asymptotic}
M.~Bramson, Y.~Lu, and B.~Prabhakar.
\newblock Asymptotic independence of queues under randomized load balancing.
\newblock {\em Queueing Systems}, 71(3):247--292, 2012.

\bibitem{dean2008mapreduce}
J.~Dean and S.~Ghemawat.
\newblock Mapreduce: simplified data processing on large clusters.
\newblock {\em Communications of the ACM}, 51(1):107--113, 2008.

\bibitem{foley2001join}
R.~D. Foley and D.~R. McDonald.
\newblock Join the shortest queue: stability and exact asymptotics.
\newblock {\em Annals of Applied Probability}, pages 569--607, 2001.

\bibitem{foschini1978basic}
G.~Foschini and J.~Salz.
\newblock A basic dynamic routing problem and diffusion.
\newblock {\em IEEE Transactions on Communications}, 26(3):320--327, 1978.

\bibitem{foss1998stability}
S.~Foss and N.~Chernova.
\newblock On the stability of a partially accessible multi-station queue with
  state-dependent routing.
\newblock {\em Queueing Systems}, 29(1):55--73, 1998.

\bibitem{foster2008cloud}
I.~Foster, Y.~Zhao, I.~Raicu, and S.~Lu.
\newblock Cloud computing and grid computing 360-degree compared.
\newblock In {\em IEEE Grid Computing Environments Workshop}, pages 1--10,
  2008.

\bibitem{gamarnik2018delay}
D.~Gamarnik, J.~N. Tsitsiklis, and M.~Zubeldia.
\newblock Delay, memory, and messaging tradeoffs in distributed service
  systems.
\newblock {\em Stochastic Systems}, 8(1):45--74, 2018.

\bibitem{gandhi2015duet}
R.~Gandhi, H.~H. Liu, Y.~C. Hu, G.~Lu, J.~Padhye, L.~Yuan, and M.~Zhang.
\newblock Duet: Cloud scale load balancing with hardware and software.
\newblock {\em ACM SIGCOMM Computer Communication Review}, 44(4):27--38, 2015.

\bibitem{garg2018migrating}
S.~K. Garg, J.~Lakshmi, and J.~Johny.
\newblock Migrating vm workloads to containers: Issues and challenges.
\newblock In {\em 2018 IEEE 11th International Conference on Cloud Computing
  (CLOUD)}, pages 778--785. IEEE, 2018.

\bibitem{georgiadis2006resource}
L.~Georgiadis, M.~J. Neely, and L.~Tassiulas.
\newblock Resource allocation and cross-layer control in wireless networks.
\newblock {\em Foundations and Trends{\textregistered} in Networking},
  1(1):1--144, 2006.

\bibitem{govindan2016evolve}
R.~Govindan, I.~Minei, M.~Kallahalla, B.~Koley, and A.~Vahdat.
\newblock Evolve or die: High-availability design principles drawn from googles
  network infrastructure.
\newblock In {\em {ACM SIGCOMM}}, pages 58--72, 2016.

\bibitem{gupta2007analysis}
V.~Gupta, M.~Harchol-Balter, K.~Sigman, and W.~Whitt.
\newblock Analysis of join-the-shortest-queue routing for web server farms.
\newblock {\em Performance Evaluation}, 64(9-12):1062--1081.

\bibitem{kannan2018proctor}
R.~S. Kannan, A.~Jain, M.~A. Laurenzano, L.~Tang, and J.~Mars.
\newblock Proctor: Detecting and investigating interference in shared
  datacenters.
\newblock In {\em 2018 IEEE International Symposium on Performance Analysis of
  Systems and Software (ISPASS)}, pages 76--86. IEEE, 2018.

\bibitem{lu2011join}
Y.~Lu, Q.~Xie, G.~Kliot, A.~Geller, J.~R. Larus, and A.~Greenberg.
\newblock Join-idle-queue: A novel load balancing algorithm for dynamically
  scalable web services.
\newblock {\em Performance Evaluation}, 68(11):1056--1071, 2011.

\bibitem{maguluri2014heavy}
S.~T. Maguluri, R.~Srikant, and L.~Ying.
\newblock Heavy traffic optimal resource allocation algorithms for cloud
  computing clusters.
\newblock {\em Performance Evaluation}, 81:20--39, 2014.

\bibitem{youtube_lec}
T.~McMullen.
\newblock {Load Balancing is Impossible.}
\newblock \url{https://www.youtube.com/watch?v=kpvbOzHUakA}, published on May
  22, 2016.

\bibitem{mitzenmacher2016analyzing}
M.~Mitzenmacher.
\newblock Analyzing distributed join-idle-queue: A fluid limit approach.
\newblock In {\em IEEE Allerton}, pages 312--318, 2016.

\bibitem{mitzenmacher2002load}
M.~Mitzenmacher, B.~Prabhakar, and D.~Shah.
\newblock Load balancing with memory.
\newblock In {\em The 43rd Annual IEEE Symposium on Foundations of Computer
  Science, 2002. Proceedings.}, pages 799--808. IEEE, 2002.

\bibitem{mukhopadhyay2016randomized}
A.~Mukhopadhyay, A.~Karthik, and R.~R. Mazumdar.
\newblock Randomized assignment of jobs to servers in heterogeneous clusters of
  shared servers for low delay.
\newblock {\em Stochastic Systems}, 6(1):90--131, 2016.

\bibitem{neely2012stability}
M.~J. Neely.
\newblock Stability and probability 1 convergence for queueing networks via
  lyapunov optimization.
\newblock {\em Journal of Applied Mathematics}, Volume 2012:Article ID 831909,
  35 pages, 2012.

\bibitem{citrix}
NetScaler.
\newblock {NetScaler 11.1 and the Least Connection Method.}
\newblock
  \url{https://docs.citrix.com/en-us/netscaler/11-1/load-balancing/load-balancing-customizing-algorithms/leastconnection-method.html},
  published on January 6, 2019.

\bibitem{nishtala2013scaling}
R.~Nishtala, H.~Fugal, S.~Grimm, M.~Kwiatkowski, H.~Lee, H.~C. Li, R.~McElroy,
  M.~Paleczny, D.~Peek, P.~Saab, D.~Stafford, T.~Tung, and V.~Venkataramani.
\newblock Scaling memcache at facebook.
\newblock In {\em NSDI}, 2013.

\bibitem{ngynx_po2}
I.~Owen Garrett~of NGINX.
\newblock {NGINX and the ``Power of Two Choices'' Load-Balancing Algorithm.}
\newblock
  \url{https://www.nginx.com/blog/nginx-power-of-two-choices-load-balancing-algorithm},
  published on November 12, 2018.

\bibitem{perry2014fastpass}
J.~Perry, A.~Ousterhout, H.~Balakrishnan, D.~Shah, and H.~Fugal.
\newblock Fastpass: A centralized zero-queue datacenter network.
\newblock In {\em ACM SIGCOMM Computer Communication Review}, volume~44, pages
  307--318. ACM, 2014.

\bibitem{shah2002use}
D.~Shah and B.~Prabhakar.
\newblock The use of memory in randomized load balancing.
\newblock In {\em IEEE International Symposium on Information Theory}, page
  125, 2002.

\bibitem{stolyar2015pull}
A.~L. Stolyar.
\newblock Pull-based load distribution in large-scale heterogeneous service
  systems.
\newblock {\em Queueing Systems}, 80(4):341--361, 2015.

\bibitem{stolyar2017pull}
A.~L. Stolyar.
\newblock Pull-based load distribution among heterogeneous parallel servers:
  the case of multiple routers.
\newblock {\em Queueing Systems}, 85(1-2):31--65, 2017.

\bibitem{haproxy_po2}
W.~Tarreau.
\newblock {HAProxy. Test Driving “Power of Two Random Choices” Load
  Balancing.}
\newblock \url{https://www.haproxy.com/blog/power-of-two-load-balancing/},
  published on February 15, 2019.

\bibitem{van2017load}
M.~van~der Boor, S.~Borst, and J.~van Leeuwaarden.
\newblock Load balancing in large-scale systems with multiple dispatchers.
\newblock In {\em IEEE INFOCOM}, 2017.

\bibitem{van2019hyper}
M.~van~der Boor, S.~Borst, and J.~van Leeuwaarden.
\newblock Hyper-scalable jsq with sparse feedback.
\newblock {\em Proceedings of the ACM on Measurement and Analysis of Computing
  Systems}, 2019.

\bibitem{weber1978optimal}
R.~R. Weber.
\newblock On the optimal assignment of customers to parallel servers.
\newblock {\em Journal of Applied Probability}, 15(2):406--413, 1978.

\bibitem{winston1977optimality}
W.~Winston.
\newblock Optimality of the shortest line discipline.
\newblock {\em Journal of Applied Probability}, 14(1):181--189, 1977.

\bibitem{ying2017power}
L.~Ying, R.~Srikant, and X.~Kang.
\newblock The power of slightly more than one sample in randomized load
  balancing.
\newblock {\em Mathematics of Operations Research}, 42(3):692--722, 2017.

\bibitem{zhou2017designing}
X.~Zhou, F.~Wu, J.~Tan, Y.~Sun, and N.~Shroff.
\newblock Designing low-complexity heavy-traffic delay-optimal load balancing
  schemes: Theory to algorithms.
\newblock {\em ACM POMACS}, 1(2):39, 2017.

\end{thebibliography}


\appendix

\section{Proofs of lemmas}

This Appendix provides the proofs of the various lemmas that we employed towards establishing our theoretical results. 


\subsection{Proof of Lemma \ref{lem:JSQ_wr}}\label{app:lem:JSQ_wr}

First, by definition, $$\sum_{i=1}^n a_i^{JSQ}(t) = \sum_{i=1}^n a_i^{WR}(t)=a(t).$$
Therefore, both $\set{a_i^{JSQ}(t)}_{i=1}^n$ and $\set{a_i^{WR}(t)}_{i=1}^n$ are feasible solutions to the optimization problem given by
\begin{equation}
\begin{aligned}
& \underset{x}{\text{minimize}}
& & \sum_{i=1}^n x_i(t) Q_i(t) \\
& \text{subject to}
& & \sum_{i=1}^n x_i(t) = a(t), \quad x_i(t) \ge 0 \: \forall i \in N
\end{aligned}
\end{equation}
The optimal solution \new{value} to this problem is simply
$$a(t)\min_i\set{Q_i(t)},$$
which is exactly the way $JSQ$ policy operates. That is, $$\sum_{i=1}^n a_i^{JSQ}(t) Q_i(t) = a(t)\min_i\set{Q_i(t)}.$$ Clearly, any other feasible solution, \eg $\set{a_i^{WR}(t)}_{i=1}^n$, cannot be better. This concludes the proof. \qed


\subsection{Proof of Lemma \ref{lem:diff_w_JSQ}}\label{app:lem:diff_w_JSQ}

Expanding the term $\sum_{i=1}^n Q_i(t)\bp{a_i(t)-a_i^{JSQ}(t)}$ yields
\begin{equation} \small \sum_{i=1}^n Q_i(t)\bp{a_i(t)-a_i^{JSQ}(t)} = 
  \sum_{i=1}^n \sum_{j=1}^m Q_i(t)\bp{a_i^j(t)-a_i^{j,JSQ}(t)}. 
\end{equation}
We now substitute $Q_i(t)$ by $Q_i(t)-\tilde{Q}_i^j(t)+\tilde{Q}_i^j(t)$. This yields
\begin{equation}\label{eq:subtitude_w}
\begin{split}
  &\sum_{i=1}^n Q_i(t)\bp{a_i(t)-a_i^{JSQ}(t)} = \cr
  & \sum_{i=1}^n \sum_{j=1}^m \bp{Q_i(t)-\tilde{Q}_i^j(t)+\tilde{Q}_i^j(t)}\bp{a_i^j(t)-a_i^{j,JSQ}(t)}. 
\end{split}
\end{equation}
We proceed with the following lemma:
\begin{lemma}\label{lem:our_JSQ}
For all $t$ it holds that 
$$\sum_{i=1}^n \sum_{j=1}^m \tilde{Q}_i^j(t)\bp{a_i^j(t)-a_i^{j,JSQ}(t)} \le 0.$$
\end{lemma}
\begin{proof}
See Appendix \ref{app:lem:our_JSQ}. 
\end{proof}

Using Lemma \ref{lem:our_JSQ} in \eqref{eq:subtitude_w} yields
\begin{equation}\label{eq:subtitude_w_2}
\begin{split}
  &\sum_{i=1}^n Q_i(t)\bp{a_i(t)-a_i^{JSQ}(t)} \le \cr
  & \sum_{i=1}^n \sum_{j=1}^m \bp{Q_i(t)-\tilde{Q}_i^j(t)}\bp{a_i^j(t)-a_i^{j,JSQ}(t)}. 
\end{split}
\end{equation}
Now, using the fact that $xy \le |x||y|$ for all $(x,y) \in \mathbb{R}^2$ on \eqref{eq:subtitude_w_2} yields
\begin{equation}\label{eq:subtitude_w_3}
\begin{split}
  &\sum_{i=1}^n Q_i(t)\bp{a_i(t)-a_i^{JSQ}(t)} \le \cr
  & \sum_{i=1}^n \sum_{j=1}^m \Big|Q_i(t)-\tilde{Q}_i^j(t)\Big| \Big|a_i^j(t)-a_i^{j,JSQ}(t)\Big|.
\end{split}
\end{equation}
Finally, it trivially holds that
\begin{equation}\label{eq:subtitude_w_4}
a(t) \ge \Big|a_i^j(t)-a_i^{j,JSQ}(t)\Big|.
\end{equation}
Using \eqref{eq:subtitude_w_4} in \eqref{eq:subtitude_w_3} concludes the proof. \qed


\subsection{Proof of Lemma \ref{lem:wr_stable}}\label{app:lem:wr_stable}

Each dispatcher applies the WR policy independently. Therefore, by applying \eqref{eq:service_1}, \eqref{eq:service_2}, \eqref{eq:arrival_1} and \eqref{eq:arrival_2} we have that the expected number of jobs arriving at each server $i$ is  
\begin{equation}\label{eq:wr_ex}
\begin{split}
    &\E \Big[a_i^{WR}(t)\Big] = \E \Big[\sum_{j=1}^m a_i^{j,WR}(t)\Big] = \cr &\frac{\mu_i^{(1)}}{\sum_{i=1}^n \mu_i^{(1)}} \E \Big[\sum_{j=1}^m a^j(t) \Big] = \frac{\lambda^{(1)}\mu_i^{(1)}}{\sum_{i=1}^n \mu_i^{(1)}}.
\end{split}
\end{equation}
Using \eqref{eq:service_1}, \eqref{eq:service_2}, \eqref{eq:subcritical_assumption} and \eqref{eq:wr_ex} yields
\begin{equation}
\begin{split}
    &\E \Big[s_i(t)-a_i^{WR}(t)\Big] = \mu_i^{(1)} - \frac{\lambda^{(1)}\mu_i^{(1)}}{\sum_{i=1}^n \mu_i^{(1)}} = \cr
    &\mu_i^{(1)} \frac{\sum_{i=1}^n \mu_i^{(1)}-\lambda^{(1)}}{\sum_{i=1}^n \mu_i^{(1)}} =  \frac{\epsilon\mu_i^{(1)}}{\sum_{i=1}^n \mu_i^{(1)}}.
\end{split}
\end{equation}
This concludes the proof. \qed


\subsection{Proof of Lemma \ref{lem:our_JSQ}}\label{app:lem:our_JSQ}

Fix $j=j^*$. It is sufficient to show that 
\begin{equation}\label{eq:suff_cond_ours_JSQ}
    \sum_{i=1}^n \tilde{Q}_i^{j^*}(t)\bp{a_i^{j^*}(t)-a_i^{{j^*},JSQ}(t)} \le 0.
\end{equation}
The proof now follows similar lines to the proof of Lemma \ref{lem:JSQ_wr}.
By definition $$\sum_{i=1}^n a_i^{j^*}(t) = \sum_{i=1}^n a_i^{{j^*},JSQ}(t)=a^{j^*}(t).$$
Therefore, both $\set{a_i^{j^*}(t)}_{i=1}^n$ and $\set{a_i^{{j^*},JSQ}(t)}_{i=1}^n$ are feasible solutions to the optimization problem given by
\begin{equation}
\begin{aligned}
& \underset{x}{\text{minimize}}
& & \sum_{i=1}^n x_i(t) \tilde{Q}_i^{j^*}(t) \\
& \text{subject to}
& & \sum_{i=1}^n x_i(t) = a^{j^*}(t), \quad x_i(t) \ge 0 \: \forall i \in N
\end{aligned}
\end{equation}
The optimal solution \new{value} to this problem is simply
$a^{j^*}(t)\min_i\set{\tilde{Q}_i^{j^*}(t)},$
which is exactly the way our policy operates since it performs $JSQ$ considering $\set{\tilde{Q}_i^{j^*}(t)}$ instead of $\set{Q_i(t)}$. That is $$\sum_{i=1}^n a_i^{j^*}(t) Q_i(t) = a^{j^*}(t)\min_i\set{\tilde{Q}_i^{j^*}(t)}.$$ Any other feasible solution including $\set{a_i^{JSQ}(t)}_{i=1}^n$ cannot be better when considering $\set{\tilde{Q}_i^{j^*}(t)}$ instead of $\set{Q_i(t)}$. This proves the inequality in \eqref{eq:suff_cond_ours_JSQ} and thus concludes the proof.
\qed


\subsection{Proof of Lemma \ref{lem:gap_reccurence}}\label{app:lem:gap_reccurence}

We prove the lemma by the way of induction on $n$.

\TT{Basis.} For $t=0$ the claim trivially holds since
$$T(0) \le C_2 \le \max \set{\frac{C_1}{\epsilon}, C_2}.$$

\TT{Induction hypothesis.} Assume that 
$$T(n) \le \max \set{\frac{C_1}{\epsilon}, C_2}.$$

\TT{Inductive step.} By definition
$$T(n+1) \le (1-\epsilon) \cdot T(n) + C_1.$$
Now, using the induction hypothesis yields
\begin{equation*}
    \begin{split}
        &T(n+1) \le (1-\epsilon)\max \set{\frac{C_1}{\epsilon}, C_2} + C_1 = \cr 
        & \max \set{\frac{C_1}{\epsilon}, C_2} - \max \set{C_1, \epsilon \cdot C_2} + C_1 \le \cr 
        & \max \set{\frac{C_1}{\epsilon}, C_2}.
    \end{split}
\end{equation*}
This concludes the proof. 
\qed


\end{document}